\documentclass[a4paper]{article}%
\usepackage{amsmath}
\usepackage{amsfonts}
\usepackage{amssymb}
\usepackage{graphicx}
\usepackage{hyperref}%
\providecommand{\U}[1]{\protect\rule{.1in}{.1in}}
\newtheorem{theorem}{Theorem}

\newtheorem{corollary}[theorem]{Corollary}

\newtheorem{definition}[theorem]{Definition}

\newtheorem{lemma}[theorem]{Lemma}

\newtheorem{proposition}[theorem]{Proposition}
\newtheorem{remark}[theorem]{Remark}

\newenvironment{proof}[1][Proof]{\noindent\textbf{#1.} }{\ \rule{0.5em}{0.5em}}
\begin{document}

\title{Decay estimates for steady solutions of the Navier-Stokes equations in two
dimensions in the presence of a wall}
\author{Christoph Boeckle\thanks{Supported by the Swiss National Science Foundation
(Grant No. 200021-124403).}~\thanks{Corresponding author.}\\{\small Theoretical Physics Department}\\{\small University of Geneva, Switzerland}\\{\small christoph.boeckle@unige.ch}
\and Peter Wittwer\thanks{Supported by the Swiss National Science Foundation (Grant
No. 200021-124403).}\\{\small Theoretical Physics Department}\\{\small University of Geneva, Switzerland}\\{\small peter.wittwer@unige.ch}}
\maketitle

\begin{abstract}
Let $\omega$ be the vorticity of a stationary solution of the two-dimensional
Navier-Stokes equations with a drift term parallel to the boundary in the
half-plane $\Omega_{+}=\left\{  \left.  (x,y)\in\mathbb{R}^{2}\right\vert
~y>1\right\}  $, with zero Dirichlet boundary conditions at $y=1$ and at
infinity, and with a small force term of compact support. Then, $|xy\omega
(x,y)|$ is uniformly bounded in $\Omega_{+}$. The proof is given in a
specially adapted functional framework and the result is a key ingredient for
obtaining information on the asymptotic behavior of the velocity at infinity.

\bigskip


\end{abstract}
\tableofcontents

\section{Introduction}

In this paper we consider the steady Navier-Stokes equations in a half-plane
$\Omega_{+}=\left\{  \left.  (x,y)\in\mathbb{R}^{2}\right\vert ~y>1\right\}  $
with a drift term parallel to the boundary, a small driving force of compact
support, with zero Dirichlet boundary conditions at the boundary of the half
plane and at infinity. See \cite{Hillairet.Wittwer-Existenceofstationary2009}
and \cite{Hillairet.Wittwer-Asymptoticdescriptionof2011} for a detailed
motivation of this problem. Existence of a strong solution for this system was
proved in \cite{Hillairet.Wittwer-Existenceofstationary2009} together with a
basic bound on the decay at infinity, and the existence of weak solutions was
shown in \cite{Hillairet.Wittwer-Asymptoticdescriptionof2011}. By elliptic
regularity weak solutions are smooth, and their only possible shortcoming is
the behavior at infinity, since the boundary condition may not be satisfied
there in a pointwise sense. In
\cite{Hillairet.Wittwer-Asymptoticdescriptionof2011} it was also shown that
for small forces there is only one weak solution. This unique weak solution
therefore coincides with the strong solution and satisfies as a consequence
the boundary condition at infinity in a pointwise sense.

The aim of this paper is to provide additional information concerning the
behavior of this solution at infinity by analyzing the solution obtained in
\cite{Hillairet.Wittwer-Existenceofstationary2009} in a more stringent
functional setting. More precisely, we obtain more information on the decay
behavior of the vorticity of the flow. Bounds on vorticity as a step towards
bounds on the velocity are a classical procedure in asymptotic analysis of
fluid flows (see the seminal papers
\cite{Gilbarg.Weinberger-AsymptoticPropertiesof1974},
\cite{Gilbarg.Weinberger-Asymptoticpropertiesof1978} and
\cite{Amick-asymptoticformof1991}). In
\cite{Hillairet.Wittwer-Existenceofstationary2009} and the current work, the
equation for the vorticity is Fourier-transformed with respect to the
coordinate $x$ parallel to the wall, and then rewritten as a dynamical system
with the coordinate $y$ perpendicular to the wall playing the role of time. In
this setting information on the behavior of the vorticity at infinity is
studied by analyzing the Fourier transform at $k=0$, with $k$ the Fourier
conjugate variable of $x$. In the present work, we also control the derivative
of the Fourier transform of the vorticity, which yields more precise decay
estimates for the vorticity and the velocity field in direct space than the
ones found in \cite{Hillairet.Wittwer-Existenceofstationary2009}. Our proof is
then based on a new linear fixed point problem involving the solution obtained
in \cite{Hillairet.Wittwer-Existenceofstationary2009} and the derivative of
the vorticity with respect to $k$.

Since the original equation is elliptic, the dynamical system under
consideration contains stable and unstable modes and no spectral gap, so that
standard versions of the center manifold theorem are not sufficient to prove
existence of solutions. Functional techniques that allow to deal with such a
situation go back to \cite{Gallay-center-stablemanifoldtheorem1993} and were
adapted to the case of the Navier-Stokes equations in
\cite{Baalen-StationarySolutionsof2002} and in
\cite{Wittwer-structureofStationary2002},
\cite{Wittwer-Supplementstructureof2003}. For a general review see
\cite{Heuveline.Wittwer-ExteriorFlowsat2010}. The linearized version of the
current problem was studied in \cite{Hillairet.Wittwer-vorticityofOseen2008}.
A related problem in three dimensions was discussed in
\cite{Guo.etal-Existenceofstationary2011}.

The results of the present paper are the basis for the work described in
\cite{Boeckle.Wittwer-Asymptoticsofsolutions2011}, where we extract several
orders of an asymptotic expansion of the vorticity and the velocity field at
infinity. The asymptotic velocity field obtained this way is divergence-free
and may be used to define artificial boundary conditions of Dirichlet type
when the system of equations is restricted to a finite sub-domain to be solved
numerically. The use of asymptotic terms as artificial boundary conditions was
pioneered in \cite{Boenisch.etal-Secondorderadaptive2008} for the related
problem of an exterior flow in the whole space in two dimensions, and in
\cite{Heuveline.Wittwer-AdaptiveBoundaryConditions2010} for the case in three dimensions.

\bigskip

Let $\mathbf{x}=(x,y)$, and let $\Omega_{+}=\{\left.  (x,y)\in\mathbb{R}%
^{2}\right\vert ~y>1\}$. The model under consideration is given by the
Navier-Stokes equations with a drift term parallel to the boundary,%
\begin{align}
-\partial_{x}\boldsymbol{u}\mathbf{+}\Delta\boldsymbol{u} &  =\boldsymbol{F}%
\mathbf{+}\boldsymbol{u}\cdot\mathbf{\nabla}\boldsymbol{u}+\mathbf{\nabla
}p~,\label{eq:nssteadyforce}\\
\mathbf{\nabla}\cdot\boldsymbol{u} &  =0~,\label{eq:incompressibility}%
\end{align}
subject to the boundary conditions%
\begin{align}
\boldsymbol{u}(x,1) &  =0~,\hspace{1cm}x\in\mathbb{R}~,\label{eq:b0}\\
\lim\limits_{\mathbf{x\rightarrow\infty}}\boldsymbol{u}\mathbf{(x)} &
=0~.\label{eq:b1}%
\end{align}
The following theorem is our main result.

\begin{theorem}
\label{thm:main}For all $\boldsymbol{F}\in C_{c}^{\infty}(\Omega_{+})$ with
$\boldsymbol{F}$ sufficiently small in a sense to be defined below, there
exist a unique vector field $\boldsymbol{u}=(u,v)$ and a function $p$
satisfying the Navier-Stokes equations (\ref{eq:nssteadyforce}),
(\ref{eq:incompressibility}) in $\Omega_{+}$ subject to the boundary
conditions (\ref{eq:b0}) and (\ref{eq:b1}). Moreover, there exists a constant
$C>0$, such that $|y^{3/2}u(x,y)|+|y^{3/2}v(x,y)|+$ $|y^{3}\omega
(x,y)|+|xy\omega(x,y)|\leq C$, for all $(x,y)\in\Omega_{+}$.
\end{theorem}

\bigskip

This theorem is a consequence of Theorem$~$\ref{thm:existence} which is proved
in Section \ref{sec:existencemap}. The crucial improvement with respect to
\cite{Hillairet.Wittwer-Existenceofstationary2009} is the bound on the
function $xy\omega(x,y)$.

\bigskip

The paper is organized as follows. In Section~\ref{sec:evoleq} we rewrite
(\ref{eq:nssteadyforce}) and (\ref{eq:incompressibility}) as a dynamical
system with $y$ playing the role of time, and Fourier-transform the equations
with respect to the variable $x$. Then, in Section~\ref{sec:integraleq}, we
recall the integral equations for the vorticity discussed in
\cite{Hillairet.Wittwer-Existenceofstationary2009} and complement them by the
ones for the derivative with respect to $k$. We then introduce in
Section~\ref{sec:funframe} certain well adapted Banach spaces which encode the
information concerning the decay of the functions at infinity. Finally, in
Section~\ref{sec:existencemap}, we reformulate the problem of showing the
existence of the derivative of vorticity with respect to $k$ as the fixed
point of a continuous map, based on the existence of solutions proved in
\cite{Hillairet.Wittwer-Existenceofstationary2009}. We present in
Sections~\ref{sec:convolutiondisc} and \ref{sec:bounds} the proofs of the
lemmas used in Section~\ref{sec:existencemap}. In the appendix, we recall
results from \cite{Hillairet.Wittwer-Existenceofstationary2009} which are
needed here.

\section{Reduction to an evolution equation\label{sec:evoleq}}

We recall the procedure used in
\cite{Hillairet.Wittwer-Existenceofstationary2009} to frame the Navier-Stokes
equations for the studied case as a dynamical system. Let $\boldsymbol{u}%
=(u,v)$ and $\boldsymbol{F}=(F_{1},F_{2})$. Then, equations
(\ref{eq:nssteadyforce}) and (\ref{eq:incompressibility}) are equivalent to%
\begin{align}
\omega &  =-\partial_{y}u+\partial_{x}v~,\label{eq:vorticity}\\
-\partial_{x}\omega\mathbf{+}\Delta\omega &  =\partial_{x}(u\omega
)+\partial_{y}(v\omega)+\partial_{x}F_{2}-\partial_{y}F_{1}%
~,\label{eq:vorticityNScomovingForce}\\
\partial_{x}u+\partial_{y}v  &  =0~. \label{eq:incompressibility2}%
\end{align}
The function $\omega$ is the vorticity of the fluid. Once equations
(\ref{eq:vorticity})-(\ref{eq:incompressibility2}) are solved, the pressure
$p$ can be obtained by solving the equation%
\[
\Delta p=-\mathbf{\nabla}\cdot(\boldsymbol{F}\mathbf{+}\boldsymbol{u}%
\cdot\mathbf{\nabla}\boldsymbol{u})
\]
in $\Omega_{+}$, subject to the Neumann boundary condition%
\[
\partial_{y}p(x,1)=\partial_{y}^{2}v(x,1)~.
\]
Let%
\begin{align}
q_{0}  &  =u\omega~,\label{eq:defq0direct}\\
q_{1}  &  =v\omega~, \label{eq:defq1direct}%
\end{align}
and let furthermore%
\begin{align}
Q_{0}  &  =q_{0}+F_{2}~,\label{eq:defQ0direct}\\
Q_{1}  &  =q_{1}-F_{1}~. \label{eq:defQ1direct}%
\end{align}
We then rewrite the second order differential equation
(\ref{eq:vorticityNScomovingForce}) as a first order system
\begin{align}
\partial_{y}\omega &  =\partial_{x}\eta+Q_{1}~,\label{eq:diffomegadirect}\\
\partial_{y}\eta &  =-\partial_{x}\omega+\omega+Q_{0}~.
\label{eq:diffetadirect}%
\end{align}
Note that, unlike the right-hand side of (\ref{eq:vorticityNScomovingForce}),
the expressions for $Q_{0}$ and $Q_{1}$ do not contain derivatives. This is
due to the fact that, in contrast to standard practice, we did not set, say,
$\partial_{y}\omega=\eta$, but we chose with (\ref{eq:diffomegadirect}) a more
sophisticated definition. The fact that the nonlinear terms in
(\ref{eq:diffomegadirect}), (\ref{eq:diffetadirect}) do not contain
derivatives simplifies the analysis of the equations considerably. An
additional trick allows to reduce complexity even further. Namely, we can
replace (\ref{eq:incompressibility2}) and (\ref{eq:vorticity}) with the
equations
\begin{align}
\partial_{y}\psi &  =-\partial_{x}\varphi-Q_{1}~,\label{eq:diffpsidirect}\\
\partial_{y}\varphi &  =\partial_{x}\psi+Q_{0}~, \label{eq:diffphidirect}%
\end{align}
if we use the decomposition%
\begin{align}
u  &  =-\eta+\varphi~,\label{eq:ansatzUdirect}\\
v  &  =\omega+\psi~. \label{eq:ansatzVdirect}%
\end{align}
The point is that in contrast to $u$ and $v$ the functions $\psi$ and
$\varphi$ decouple on the linear level from $\omega$ and $\eta$. Since, on the
linear level we have $\Delta\varphi=0$ and $\Delta\psi=0$, it will turn out
that $\varphi$ and $\psi$ have a dominant asymptotic behavior which is
harmonic when $Q_{0}$ and $Q_{1}$ are small.

Equations (\ref{eq:diffomegadirect})-(\ref{eq:diffphidirect}) are a dynamical
system with $y$ playing the role of time. We now take the Fourier transform in
the $x$-direction.

\begin{definition}
\label{def:fourier}Let $\hat{f}$ be a complex valued function on $\Omega_{+}$.
Then, we define the inverse Fourier transform $f=\mathcal{F}^{-1}[\hat{f}]$ by
the equation,
\[
f(x,y)=\mathcal{F}^{-1}[\hat{f}](x,y)=\frac{1}{2\pi}\int_{\mathbb{R}}%
e^{-ikx}\hat{f}(k,y)dk~,
\]
and $\hat{h}=\hat{f}\ast\hat{g}$ by%
\[
\hat{h}(k,y)=(\hat{f}\ast\hat{g})(k,y)=\frac{1}{2\pi}\int_{\mathbb{R}}\hat
{f}(k-k^{\prime},y)\hat{g}(k^{\prime},y)dk^{\prime}~,
\]
whenever the integrals make sense. We note that for a function $f$ which is
smooth and of compact support in $\Omega_{+}$ we have $f=\mathcal{F}^{-1}%
[\hat{f}]$, where%
\[
\hat{f}(k,y)=\mathcal{F}[f](k,y)=\int_{\mathbb{R}}e^{ikx}f(x,y)dx~,
\]
and that $fg=\mathcal{F}^{-1}[\hat{f}\ast\hat{g}]$.
\end{definition}

With these definitions we have in Fourier space, instead of
(\ref{eq:diffomegadirect})-(\ref{eq:diffphidirect}), the equations%
\begin{align}
\partial_{y}\hat{\omega}  &  =-ik\hat{\eta}+\hat{Q}_{1}%
~,\label{eq:diffomegafourier}\\
\partial_{y}\hat{\eta}  &  =(ik+1)\hat{\omega}+\hat{Q}_{0}%
~,\label{eq:diffetafourier}\\
\partial_{y}\hat{\psi}  &  =ik\hat{\varphi}-\hat{Q}_{1}%
~,\label{eq:diffpsifourier}\\
\partial_{y}\hat{\varphi}  &  =-ik\hat{\psi}+\hat{Q}_{0}~.
\label{eq:diffphifourier}%
\end{align}
From (\ref{eq:defQ0direct}) and (\ref{eq:defQ1direct}) we get%
\begin{align}
\hat{Q}_{0}  &  =\hat{q}_{0}+\hat{F}_{2}~,\label{eq:defQ0fourier}\\
\hat{Q}_{1}  &  =\hat{q}_{1}-\hat{F}_{1}~, \label{eq:defQ1fourier}%
\end{align}
from (\ref{eq:defq0direct}) and (\ref{eq:defq1direct}) we get%
\begin{align}
\hat{q}_{0}  &  =\hat{u}\ast\hat{\omega}~,\label{eq:defq0fourier}\\
\hat{q}_{1}  &  =\hat{v}\ast\hat{\omega}~, \label{eq:defq1fourier}%
\end{align}
and instead of (\ref{eq:ansatzUdirect}) and (\ref{eq:ansatzVdirect}) we have
the equations%
\begin{align}
\hat{u}  &  =-\hat{\eta}+\hat{\varphi}~,\label{eq:ansatzUfourier}\\
\hat{v}  &  =\hat{\omega}+\hat{\psi}~. \label{eq:ansatzVfourier}%
\end{align}

\section{Integral equations\label{sec:integraleq}}

We now reformulate the problem of finding a solution to
(\ref{eq:diffomegafourier})-(\ref{eq:diffphifourier}) which satisfies the
boundary conditions (\ref{eq:b0}) and (\ref{eq:b1}) in terms of a system of
integral equations. The equations for $\hat{\omega}$, $\hat{\eta}$,
$\hat{\varphi}$ and $\hat{\psi}$ are as in
\cite{Hillairet.Wittwer-Existenceofstationary2009}. In particular we recall
that%
\begin{equation}
\hat{\omega}=\sum_{m=0}^{1}\sum_{n=1}^{3}\hat{\omega}_{n,m}~,
\label{eq:omegaInt}%
\end{equation}
where, for $n=1,2,3$, $m=0,1$,
\begin{equation}
\hat{\omega}_{n,m}(k,t)=\check{K}_{n}(k,t-1)\int_{I_{n}}\check{f}%
_{n,m}(k,s-1)\hat{Q}_{m}(k,s)ds~, \label{eq:omegaIntComp}%
\end{equation}
where, for $k\in\mathbb{R}\setminus\{0\}$ and $\sigma$, $\tau\geq0$,%
\begin{align}
\check{K}_{n}(k,\tau)  &  =\frac{1}{2}e^{-\kappa\tau}~,\text{ for
}n=1,2~,\label{Kn}\\
\check{K}_{3}(k,\tau)  &  =\frac{1}{2}(e^{\kappa\tau}-e^{-\kappa\tau})~,
\label{K3}%
\end{align}
and%
\begin{align}
\check{f}_{1,0}(k,\sigma)  &  =\frac{ik}{\kappa}e^{\kappa\sigma}-\frac{\left(
|k|+\kappa\right)  ^{2}}{\kappa}e^{-\kappa\sigma}+2\left(  |k|+\kappa\right)
e^{-|k|\sigma}~,\label{eq:f10}\\
\check{f}_{2,0}(k,\sigma)  &  =2\left(  \kappa+|k|\right)  \left(
e^{-|k|\sigma}-e^{-\kappa\sigma}\right)  ~,\label{eq:f20}\\
\check{f}_{3,0}(k,\sigma)  &  =\frac{ik}{\kappa}e^{-\kappa\sigma
}~,\label{eq:f30}\\
\check{f}_{1,1}(k,\sigma)  &  =e^{\kappa\sigma}+\frac{\left(  |k|+\kappa
\right)  ^{2}}{ik}e^{-\kappa\sigma}-2\frac{|k|\left(  |k|+\kappa\right)  }%
{ik}e^{-|k|\sigma}~,\label{eq:f11}\\
\check{f}_{2,1}(k,\sigma)  &  =2\left(  \frac{|k|\left(  |k|+\kappa\right)
}{ik}-1\right)  e^{-\kappa\sigma}-2\frac{|k|\left(  |k|+\kappa\right)  }%
{ik}e^{-|k|\sigma}~,\label{eq:f21}\\
\check{f}_{3,1}(k,\sigma)  &  =-e^{-\kappa\sigma}~, \label{eq:f31}%
\end{align}
and where $I_{1}=[1,t]$ and $I_{2}=I_{3}=[t,\infty)$.

We introduce the integral equation for $\partial_{k}\hat{\omega}$, noting that
$\hat{\omega}$ is continuous at $k=0$ (see
\cite{Hillairet.Wittwer-Existenceofstationary2009}). From
(\ref{eq:omegaIntComp}) we get that
\begin{equation}
\partial_{k}\hat{\omega}=\sum_{m=0}^{1}\sum_{n=1}^{3}\sum_{l=1}^{3}%
\partial_{k}\hat{\omega}_{l,n,m}~,\label{eq:dkomegaInt}%
\end{equation}
where, for $n=1,2,3$, $m=0,1$,%
\begin{align}
\partial_{k}\hat{\omega}_{1,n,m}(k,t) &  =\partial_{k}\check{K}_{n}%
(k,t-1)\int_{I_{n}}\check{f}_{n,m}(k,s-1)\hat{Q}_{m}%
(k,s)ds~,\label{eq:dkomegaIntComp1}\\
\partial_{k}\hat{\omega}_{2,n,m}(k,t) &  =\check{K}_{n}(k,t-1)\int_{I_{n}%
}\partial_{k}\check{f}_{n,m}(k,s-1)\hat{Q}_{m}%
(k,s)ds~,\label{eq:dkomegaIntComp2}\\
\partial_{k}\hat{\omega}_{3,n,m}(k,t) &  =\check{K}_{n}(k,t-1)\int_{I_{n}%
}\check{f}_{n,m}(k,s-1)\partial_{k}\hat{Q}_{m}%
(k,s)ds~,\label{eq:dkomegaIntComp3}%
\end{align}
where, for $k\in\mathbb{R}\setminus\{0\}$ and $\sigma$, $\tau\geq0$,%
\begin{align}
\partial_{k}\check{K}_{n}(k,\tau) &  =\frac{1}{4}\dfrac{2k-i}{\kappa
}e^{-\kappa\tau}~,\text{ for }n=1,2~,\label{dKn}\\
\partial_{k}\check{K}_{3}(k,\tau) &  =\frac{1}{4}\dfrac{2k-i}{\kappa
}(e^{\kappa\tau}+e^{-\kappa\tau})~,\label{dK3}%
\end{align}
where $\check{f}_{n,m}$ is as above, where
\begin{align}
\partial_{k}\check{f}_{1,0}(k,\sigma) &  =\frac{i}{2\kappa}(e^{\kappa\sigma
}+e^{-\kappa\sigma}-2e^{-|k|\sigma})-\frac{ik^{2}}{2\kappa^{3}}(e^{\kappa
\sigma}-e^{-\kappa\sigma})+\frac{2}{\kappa}\frac{k^{2}+|k|\kappa}%
{k}(e^{-|k|\sigma}-e^{-\kappa\sigma})\nonumber\\
&  +i\frac{k^{2}+\kappa^{2}}{2\kappa^{2}}(e^{\kappa\sigma}-e^{-\kappa\sigma
})\sigma+\frac{k^{2}+|k|\kappa}{k}\frac{k^{2}+\kappa^{2}}{\kappa^{2}%
}e^{-\kappa\sigma}\sigma-2\frac{k^{2}+|k|\kappa}{k}e^{-|k|\sigma}%
\sigma~,\label{eq:dkf10}\\
\partial_{k}\check{f}_{2,0}(k,\sigma) &  =\frac{(|k|+\kappa)^{2}}{\kappa
k}(e^{-|k|\sigma}-e^{-\kappa\sigma})-2\frac{\kappa+|k|}{\kappa k}\left(
|k|\kappa e^{-|k|\sigma}-\frac{k^{2}+\kappa^{2}}{2}e^{-\kappa\sigma}\right)
\sigma~,\label{eq:dkf20}\\
\partial_{k}\check{f}_{3,0}(k,\sigma) &  =\frac{k}{2\kappa^{3}}e^{-\kappa
\sigma}-i\frac{k^{2}+\kappa^{2}}{2\kappa^{2}}\sigma e^{-\kappa\sigma
}~,\label{eq:dkf30}\\
\partial_{k}\check{f}_{1,1}(k,\sigma) &  =i\frac{\left(  |k|+\kappa\right)
^{2}}{\kappa|k|}(e^{-|k|\sigma}-e^{-\kappa\sigma})\nonumber\\
&  +\frac{k^{2}+\kappa^{2}}{2\kappa k}(e^{\kappa\sigma}+e^{-\kappa\sigma
})\sigma+2i\frac{k^{2}+|k|\kappa}{k^{2}}\left(  \frac{k^{2}+\kappa^{2}%
}{2\kappa}e^{-\kappa\sigma}-|k|e^{-|k|\sigma}\right)  \sigma~,\label{eq:dkf11}%
\\
\partial_{k}\check{f}_{2,1}(k,\sigma) &  =i\frac{\left(  |k|+\kappa\right)
^{2}}{\kappa|k|}(e^{-\kappa\sigma}-e^{-|k|\sigma})+i(|k|+\kappa)\frac
{k^{2}+\kappa^{2}}{k^{2}}e^{-\kappa\sigma}\sigma-2i\left(  |k|+\kappa\right)
e^{-|k|\sigma}\sigma~,\label{eq:dkf21}\\
\partial_{k}\check{f}_{3,1}(k,\sigma) &  =\frac{k^{2}+\kappa^{2}}{2\kappa
k}e^{-\kappa\sigma}\sigma~.\label{eq:dkf31}%
\end{align}
and where the functions
\begin{align*}
\partial_{k}\hat{Q}_{0} &  =\partial_{k}\hat{q}_{0}+\partial_{k}\hat{F}%
_{2}~,\\
\partial_{k}\hat{Q}_{1} &  =\partial_{k}\hat{q}_{1}-\partial_{k}\hat{F}_{1}~,
\end{align*}
are obtained from (\ref{eq:defQ0fourier})\ and (\ref{eq:defQ1fourier}). Since
$\hat{q}_{0}$ and $\hat{q}_{1}$ are convolution products (see
(\ref{eq:defq0fourier}) and (\ref{eq:defq1fourier})), and noting that $\hat
{u}$ and $\hat{v}$ are continuous bounded functions on $\mathbb{R}$, that
$\hat{\omega}$ is continuous on $\mathbb{R}$ and differentiable on
$\mathbb{R}\setminus\{0\}$ and that $\partial_{k}\hat{\omega}$ is absolutely
integrable, we conclude (see \cite[Proposition 8.8, page 241]{Folland1999})
that $\hat{q}_{0}$ and $\hat{q}_{1}$ are continously differentiable functions
and that
\begin{align}
\partial_{k}\hat{q}_{0} &  =\hat{u}\ast\partial_{k}\hat{\omega}%
~,\label{eq:defdkq0}\\
\partial_{k}\hat{q}_{1} &  =\hat{v}\ast\partial_{k}\hat{\omega}%
~.\label{eq:defdkq1}%
\end{align}
This means that it is sufficient to add equation (\ref{eq:dkomegaInt}) to the
ones for $\hat{\omega}$, $\hat{\eta}$, $\hat{\varphi}$ and $\hat{\psi}$ in
order to get a set of integrals equations determining also $\partial_{k}%
\hat{\omega}$.

\begin{remark}
\label{rem:kcheckandk}The products $\check{K}_{n}\check{f}_{n,m}$ are equal to
$K_{n}f_{n,m}$ as defined in
\cite{Hillairet.Wittwer-Existenceofstationary2009}, and we have $\check
{K}_{n=1,2}=K_{n=1,2}$, $\check{K}_{3}=\frac{ik}{\kappa}K_{3}$, $\check
{f}_{n=1,2;m}=f_{n=1,2;m}$ and $\check{f}_{3,m}=\frac{\kappa}{ik}f_{3,m}$. We
chose to rewrite the equations in the new form for convenience later on.
\end{remark}

\section{Functional framework\label{sec:funframe}}

We recall the definition of the function spaces introduced in
\cite{Hillairet.Wittwer-Existenceofstationary2009} and extend it to include
functions with a certain type of singular behavior. Let $\alpha$, $r\geq0$,
$k\in\mathbb{R}$, $t\geq1$, and let
\begin{equation}
\mu_{\alpha,r}(k,t)=\frac{1}{1+(|k|t^{r})^{\alpha}}~. \label{eq:defMu}%
\end{equation}
Let furthermore%
\begin{align*}
\bar{\mu}_{\alpha}  &  =\mu_{\alpha,1}(k,t)~,\\
\tilde{\mu}_{\alpha}  &  =\mu_{\alpha,2}(k,t)~.
\end{align*}
We also define%
\begin{equation}
\kappa=\sqrt{k^{2}-ik}~, \label{eq:defKappa}%
\end{equation}
and%
\begin{equation}
\Lambda_{-}=-\operatorname{Re}(\kappa)=-\frac{1}{2}\sqrt{2\sqrt{k^{2}+k^{4}%
}+2k^{2}}~. \label{LM}%
\end{equation}
Throughout this paper we use the inequalities%
\begin{equation}
|\kappa|=(k^{2}+k^{4})^{1/4}\leq|k|^{1/2}+|k|\leq2^{3/4}|\kappa|\leq
2^{3/4}(1+|k|)~. \label{bk1}%
\end{equation}
We have in particular that
\begin{equation}
|k|^{\frac{1}{2}}\leq\mathrm{const.}|\kappa|~, \label{eq:sqrtKleqkappa}%
\end{equation}
and that%
\begin{equation}
e^{\Lambda_{-}\sigma}\leq e^{-|k|\sigma}~, \label{expbound}%
\end{equation}
which will play a crucial role for small and large values of $k$, respectively.

\begin{definition}
\label{def:Bnapq}We define, for fixed $\alpha\geq0$, and $n$, $p$, $q$ $\geq
0$, $\mathcal{B}_{\alpha,p,q}^{n}$ to be the Banach space of functions
$\hat{f}\colon\mathbb{R}\setminus\{0\}\times\lbrack1,\infty)\rightarrow
\mathbb{C}$, for which $\hat{f}_{n}=\kappa^{n}\cdot\hat{f}\in C(\mathbb{R}%
\setminus\{0\}\times\lbrack1,\infty),\mathbb{C})$, and for which the norm%
\[
\left\Vert \hat{f};\mathcal{B}_{\alpha,p,q}^{n}\right\Vert =\sup_{t\geq1}%
\sup_{k\in\mathbb{R}\setminus\{0\}}\frac{|\hat{f}_{n}(k,t)|}{\frac{1}{t^{p}%
}\bar{\mu}_{\alpha}(k,t)+\frac{1}{t^{q}}\tilde{\mu}_{\alpha}(k,t)}%
\]
is finite. We use the shorthand $\mathcal{B}_{\alpha,p,q}$\ for $\mathcal{B}%
_{\alpha,p,q}^{0}$. Furthermore we set, for $\alpha>2$,%
\begin{align*}
\mathcal{D}_{\alpha-1,p,q}^{1} &  =\mathcal{B}_{\alpha,p,q}^{1}\times
\mathcal{B}_{\alpha-\frac{1}{2},p+\frac{1}{2},q+\frac{1}{2}}^{1}%
\times\mathcal{B}_{\alpha-1,p+\frac{1}{2},q+1}^{1}~,\\
\mathcal{V}_{\alpha} &  =\mathcal{B}_{\alpha,\frac{5}{2},1}\times
\mathcal{B}_{\alpha,\frac{1}{2},0}\times\mathcal{B}_{\alpha,\frac{1}{2},1}~.
\end{align*}

\end{definition}

\begin{remark}
We present two elementary properties of the spaces $\mathcal{B}_{\alpha
,p,q}^{n}$, which will be routinely used without mention. Let \textit{
}$\alpha$\textit{, }$\alpha^{\prime}\geq0$\textit{, and }$p$\textit{,
}$p^{\prime}$\textit{, }$q$\textit{, }$q^{\prime}$\textit{ }$\geq0$\textit{,
then }%
\[
\mathcal{B}_{\alpha,p,q}^{n}\cap\mathcal{B}_{\alpha^{\prime},p^{\prime
},q^{\prime}}^{n}\subset\mathcal{B}_{\min\{\alpha^{\prime},\alpha
,\},\min\{p^{\prime},p\},\min\{q^{\prime},q\}}^{n}~.
\]
In addition\textit{ we have}%
\[
\mathcal{B}_{\alpha,p,q}^{n}\subset\mathcal{B}_{\alpha,\min\{p,q\},\infty}%
^{n}~,
\]
where the space with $q=\infty$ is to be understood to contain functions for
which the norm%
\[
\left\Vert \hat{f};\mathcal{B}_{\alpha,p,\infty}^{n}\right\Vert =\sup_{t\geq
1}\sup_{k\in\mathbb{R}\setminus\{0\}}\frac{|\hat{f}_{n}(k,t)|}{\frac{1}{t^{p}%
}\bar{\mu}_{\alpha}(k,t)}%
\]
is finite.
\end{remark}

\section{Existence of solutions\label{sec:existencemap}}

In \cite{Hillairet.Wittwer-Existenceofstationary2009} it was shown that one
can rewrite the integral equations as a fixed point problem, and that, for
$\boldsymbol{F}$ sufficiently small, there exist functions $\hat{\omega}$,
$\hat{u}$ and $\hat{v}$, that are solution to (\ref{eq:vorticity}%
)-(\ref{eq:incompressibility2}), satisfying the boundary conditions
(\ref{eq:b0}) and (\ref{eq:b1}). More precisely, we have, for $\alpha>3$,%
\begin{align}
\hat{\omega} &  \in\mathcal{B}_{\alpha,\frac{5}{2},1}~,\label{eq:omegaspace}\\
\hat{u} &  \in\mathcal{B}_{\alpha,\frac{1}{2},0}~,\label{eq:uspace}\\
\hat{v} &  \in\mathcal{B}_{\alpha,\frac{1}{2},1}~,\label{eq:vspace}%
\end{align}
and, for $i=0,1$,%
\begin{equation}
\hat{Q}_{i}\in\mathcal{B}_{\alpha,\frac{7}{2},\frac{5}{2}}~.\label{eq:Qispace}%
\end{equation}
We now show that using this solution as a starting point, we may define a
linear fixed point problem with a unique solution for $\partial_{k}\hat
{\omega}$. The structure of (\ref{eq:dkomegaInt}) is rather complicated and it
turns out to be necessary to decompose the sum into three parts which are
analyzed independently. Let $\mathbf{\hat{d}}=(\hat{d}_{1},\hat{d}_{2},\hat
{d}_{3})$ where
\[
\hat{d}_{l}=\sum_{m=0}^{1}\sum_{n=1}^{3}\partial_{k}\hat{\omega}_{l,n,m}~,
\]
then $\partial_{k}\hat{\omega}=\sum_{l=1}^{3}\hat{d}_{l}$. The function
$\hat{d}_{3}$ depends on $\partial_{k}\hat{\omega}$, but $\hat{d}_{1}$ and
$\hat{d}_{2}$ do not.

\begin{proposition}
\label{prop:d1&d2space}The functions $\hat{d}_{1}$ and $\hat{d}_{2}$ are in
$\mathcal{D}_{\alpha-1,\frac{3}{2},0}^{1}$.
\end{proposition}

\begin{proof}
See Sections~\ref{sec:boundsd1} and \ref{sec:boundsd2}.
\end{proof}

\bigskip

We now define the fixed point problem.

\begin{lemma}
\label{lem:P}Let $\alpha>3$, and let $\hat{u}$ and $\hat{v}$ be as in
(\ref{eq:uspace}) and (\ref{eq:vspace}) respectively. Then,%
\[%
\begin{array}
[c]{cccc}%
\mathfrak{L}_{1}~\colon & \mathcal{D}_{\alpha-1,\frac{3}{2},0}^{1} &
\rightarrow & \mathcal{B}_{\alpha,\frac{3}{2},1}\times\mathcal{B}%
_{\alpha,\frac{3}{2},2}\\
& \hat{d} & \longmapsto & \left(
\begin{array}
[c]{c}%
\hat{u}\ast\hat{d}\\
\hat{v}\ast\hat{d}%
\end{array}
\right)  ~,
\end{array}
\]
defines a continuous linear map.
\end{lemma}

\begin{proof}
The map $\mathfrak{L}_{1}$ is linear by definition of the convolution
operation. Using Corollary~\ref{corr:convsimplification} we get that the map
$\mathfrak{L}_{1}$ is bounded, since%
\begin{equation}
\left\Vert \hat{u}\ast\hat{d};\mathcal{B}_{\alpha,\frac{3}{2},1}\right\Vert
\leq\mathrm{const.}\left\Vert \hat{u};\mathcal{B}_{\alpha,\frac{1}{2}%
,0}\right\Vert \cdot\left\Vert d;\mathcal{D}_{\alpha-1,\frac{3}{2},0}%
^{1}\right\Vert ~, \label{eq:lemQ1}%
\end{equation}
and%
\begin{equation}
\left\Vert \hat{v}\ast\hat{d};\mathcal{B}_{\alpha,\frac{3}{2},2}\right\Vert
\leq\mathrm{const.}\left\Vert \hat{v};\mathcal{B}_{\alpha,\frac{1}{2}%
,1}\right\Vert \cdot\left\Vert d;\mathcal{D}_{\alpha-1,\frac{3}{2},0}%
^{1}\right\Vert ~. \label{eq:lemQ2}%
\end{equation}

\end{proof}

\begin{lemma}
\label{lem:Q}Let $\alpha>3$, $\hat{d}_{3}=\sum_{m=0}^{1}\sum_{n=1}^{3}%
\partial_{k}\hat{\omega}_{3,n,m}$ and let $\partial_{k}\hat{\omega}_{3,n,m}$
be given by (\ref{eq:dkomegaIntComp3}). Then, we have
\[%
\begin{array}
[c]{cccc}%
\mathfrak{L}_{2}~\colon & \mathcal{B}_{\alpha,\frac{3}{2},1}\times
\mathcal{B}_{\alpha,\frac{3}{2},2} & \rightarrow & \mathcal{D}_{\alpha
-1,\frac{3}{2},0}^{1}\\
& \left(
\begin{array}
[c]{c}%
\partial_{k}\hat{Q}_{0}\\
\partial_{k}\hat{Q}_{1}%
\end{array}
\right)  & \longmapsto & \hat{d}_{3}~,
\end{array}
\]
which defines a continuous linear map.
\end{lemma}

\begin{proof}
The map $\mathfrak{L}_{2}$ is linear by definition of $\hat{d}_{3}$ and is
proved to be bounded in Section~\ref{sec:boundsd3}.
\end{proof}

\subsection{Proof of Theorem~\ref{thm:main}}

Theorem~\ref{thm:main} is a consequence of the following theorem.

\begin{theorem}
[Existence]\label{thm:existence} Let $\alpha>3$, $\boldsymbol{F}=(F_{1}%
,F_{2})\in C_{c}^{\infty}(\Omega_{+})$, and let $\boldsymbol{\hat{F}}=(\hat
{F}_{1},\hat{F}_{2})$ be the Fourier transform of $\boldsymbol{F}$. If
$\Vert(\hat{F}_{2},-\hat{F}_{1});\mathcal{B}_{\alpha,\frac{7}{2},\frac{5}{2}%
}\times\mathcal{B}_{\alpha,\frac{7}{2},\frac{5}{2}}\Vert$ is sufficiently
small, then there exists a unique solution $(\hat{\omega},\hat{u},\hat
{v},\mathbf{\hat{d})}$ in $\mathcal{V}_{\alpha}\times\mathcal{D}%
_{\alpha-1,\frac{3}{2},0}^{1}$.
\end{theorem}

\begin{proof}
We have the existence and uniqueness of $(\hat{\omega},\hat{u},\hat{v}%
)\in\mathcal{V}_{\alpha}$ thanks to
\cite{Hillairet.Wittwer-Existenceofstationary2009} and
\cite{Hillairet.Wittwer-Asymptoticdescriptionof2011}. Since $\alpha>3$, we
have by Lemmas~\ref{lem:P} and \ref{lem:Q} that the map $\mathfrak{C}%
~\colon\mathcal{D}_{\alpha-1,\frac{3}{2},0}^{1}\rightarrow\mathcal{D}%
_{\alpha-1,\frac{3}{2},0}^{1}$, $x\mapsto\mathfrak{C}[x]=\mathfrak{L}%
_{2}[\mathfrak{L}_{1}[\hat{d}_{1}+\hat{d}_{2}+x]+(\partial_{k}\hat{F}%
_{2},-\partial_{k}\hat{F}_{1})]$ is continuous. Since from
\cite{Hillairet.Wittwer-Existenceofstationary2009} we have that $\left\Vert
(\hat{\omega},\hat{u},\hat{v});\mathcal{V}_{\alpha}\right\Vert \leq
\mathrm{const.}\left\Vert (\hat{F}_{2},-\hat{F}_{1});\mathcal{B}_{\alpha
,\frac{7}{2},\frac{5}{2}}\times\mathcal{B}_{\alpha,\frac{7}{2},\frac{5}{2}%
}\right\Vert $, we find with (\ref{eq:lemQ1}) and (\ref{eq:lemQ2}) that the
image of $\mathfrak{L}_{1}$ is arbitrarily small. We then have by linearity of
$\mathfrak{L}_{2}$, that $\mathfrak{C}$ has a fixed point since $\left\Vert
(\partial_{k}\hat{F}_{2},-\partial_{k}\hat{F}_{1});\mathcal{B}_{\alpha
,\frac{3}{2},1}\times\mathcal{B}_{\alpha,\frac{3}{2},2}\right\Vert <\infty$.
This completes the proof of Theorem~\ref{thm:existence}.
\end{proof}

\bigskip

Theorem~\ref{thm:main} now follows by inverse Fourier transform and the decay
properties are a direct consequence of the spaces of which $\hat{u}$, $\hat
{v}$, $\hat{\omega}$ and $\partial_{k}\hat{\omega}$ are elements. Indeed, for
a function $\hat{f}\in\mathcal{B}_{\alpha,p,q}^{n}$ with $\alpha>3$, $n=0,1$
and $p$, $q$ $\geq0$, we have from the definition of the Fourier transform
that
\begin{equation}
\sup\limits_{x\in\mathbb{R}}\left\vert f(x,y)\right\vert \leq\frac{1}{2\pi
}\int_{\mathbb{R}}\left\vert \hat{f}(k,y)\right\vert
dk~,\label{eq:fouriertodirectspace}%
\end{equation}
and from the definition of the function spaces that%
\begin{align}
\int_{\mathbb{R}}\left\vert \hat{f}(k,t)\right\vert ~dk &  \leq\left\Vert
\hat{f}_{n};\mathcal{B}_{\alpha,p,q}^{n}\right\Vert \int_{\mathbb{R}}\frac
{1}{\kappa^{n}}\left(  \frac{1}{t^{p}}\bar{\mu}_{\alpha}(k,t)+\frac{1}{t^{q}%
}\tilde{\mu}_{\alpha}(k,t)\right)  dk\nonumber\\
&  \leq\mathrm{const.}\left\Vert \hat{f}_{n};\mathcal{B}_{\alpha,p,q}%
^{n}\right\Vert \left(  \frac{1}{t^{p+(1-n)}}+\frac{1}{t^{q+2(1-n)}}\right)
\nonumber\\
&  \leq\frac{\mathrm{const.}}{t^{\min\{p+(1-n),q+2(1-n))}}\left\Vert \hat
{f}_{n};\mathcal{B}_{\alpha,p,q}^{n}\right\Vert
~.\label{eq:fourierspacebounds}%
\end{align}
Combining (\ref{eq:fouriertodirectspace}) and (\ref{eq:fourierspacebounds}) we
have%
\[
\sup\limits_{x\in\mathbb{R}}\left\vert f(x,y)\right\vert \leq\frac
{\mathrm{const.}}{y^{\min\{p+(1-n),q+2(1-n))}}\left\Vert \hat{f}%
_{n};\mathcal{B}_{\alpha,p,q}^{n}\right\Vert ~.
\]
Finally, we have, using that $(\hat{\omega},\hat{u},\hat{v},\mathbf{\hat{d}%
)}\in\mathcal{V}_{\alpha}\times\mathcal{D}_{\alpha-1,\frac{3}{2},0}^{1}$~, and
that%
\[
|x\omega(x,y)|\leq\frac{1}{2\pi}\int_{\mathbb{R}}|\partial_{k}\omega(k,y)|dk~,
\]
that%
\begin{align*}
|y^{3/2}u(x,y)| &  \leq C_{1}~,~|y^{3/2}v(x,y)|\leq C_{2}~,\\
|y^{3}\omega(x,y)| &  \leq C_{3}~,~|yx\omega(x,y)|\leq C_{4}~,
\end{align*}
with $C_{i}\in\mathbb{R}$, for $i=1,\ldots,4$, which proves the bound in
Theorem~\ref{thm:main}.

\section{\label{sec:convolutiondisc}Convolution with singularities}

We first recall the convolution result from
\cite{Hillairet.Wittwer-Existenceofstationary2009}.

\begin{proposition}
[convolution]\label{prop:convHW}Let $\alpha$, $\beta>1$, and $r$, $s\geq0$ and
let $a$, $b$ be continuous functions from $\mathbb{R}_{0}\times\lbrack
1,\infty)$ to $\mathbb{C}$ satisfying the bounds,%
\begin{align*}
\left\vert a(k,t)\right\vert  &  \leq\mu_{\alpha,r}(k,t)~,\\
\left\vert b(k,t)\right\vert  &  \leq\mu_{\beta,s}(k,t)~.
\end{align*}
Then, the convolution $a\ast b$ is a continuous function from $\mathbb{R}%
\times\lbrack1,\infty)$ to $\mathbb{C}$ and we have the bound%
\[
\left\vert \left(  a\ast b\right)  (k,t)\right\vert \leq\mathrm{const.}\left(
\frac{1}{t^{r}}\mu_{\beta,s}(k,t)+\frac{1}{t^{s}}\mu_{\alpha,r}(k,t)\right)
~,
\]
uniformly in $t\geq1$, $k\in\mathbb{R}$.
\end{proposition}

\vspace{0in}

Since $\partial_{k}\hat{\omega}$ diverges like $|\kappa|^{-1}$ at $k=0$ we
need to strengthen this result.

\begin{proposition}
[convolution with $\left\vert \kappa\right\vert ^{-1}$ singularity]%
\label{prop:convwithroot} Let $\alpha,\tilde{\beta}>1$ and $r,\tilde{s}\geq0$,
let $a$ be as in Proposition~\ref{prop:convHW} and $\tilde{b}$ a continuous
function from $\mathbb{R}_{0}\times\lbrack1,\infty)$ to $\mathbb{C}$,
satisfying the bound
\[
\left\vert \tilde{b}(k,t)\right\vert \leq\left\vert \kappa(k)\right\vert
^{-1}\mu_{\tilde{\beta},\tilde{s}}\left(  k,t\right)  ~,
\]
then the convolution $a\ast\tilde{b}$ is a continuous function from
$\mathbb{R\times\lbrack}1,\infty)\rightarrow\mathbb{C}$ and we have the bounds%
\begin{align}
\left\vert (a\ast\tilde{b})(k,t)\right\vert  &  \leq\mathrm{const.}\left(
\max\left\{  \frac{1}{t^{\frac{\tilde{s}}{2}}},\frac{1}{t^{\frac{\tilde
{s}+r-\tilde{s}^{\prime}}{2}}}\right\}  \mu_{\tilde{\beta},s^{\prime}}\left(
k,t\right)  +\frac{1}{t^{\frac{\tilde{s}}{2}}}\mu_{\alpha,r}\left(
k,t\right)  \right)  ~,\label{eq:convwithroot}\\
\left\vert (a\ast\tilde{b})(k,t)\right\vert  &  \leq\mathrm{const.}\left(
\max\left\{  \frac{1}{t^{\frac{\tilde{s}}{2}}},\frac{1}{t^{r-c\tilde
{s}^{\prime}}}\right\}  \mu_{\tilde{\beta}+c,\tilde{s}^{\prime}}\left(
k,t\right)  +\frac{1}{t^{\frac{\tilde{s}}{2}}}\mu_{\alpha,r}\left(
k,t\right)  \right)  ~, \label{eq:convwithrootgainbeta}%
\end{align}
for $\tilde{s}^{\prime}\leq\tilde{s}$, and $c\in\left\{  \frac{1}%
{2},1\right\}  $.
\end{proposition}

\begin{proof}
We drop the \symbol{126}~to unburden the notation. Continuity is elementary.
Since the functions $\mu_{\alpha,r}$ are even in $k$, we only consider
$k\geq0$. The proof is in two parts, one for $0\leq k\leq t^{-s^{\prime}}$ and
the other for $t^{-s^{\prime}}<k$. The first part is valid for both
(\ref{eq:convwithroot}) and (\ref{eq:convwithrootgainbeta}). For $0\leq k\leq
t^{-s^{\prime}}$, and $\alpha^{\prime}\geq0$, we have%
\begin{align*}
\left\vert (a\ast b)(k,t)\right\vert  &  \leq\int_{\mathbb{R}}\mu_{\alpha
,r}(k^{\prime},t)|\kappa(k-k^{\prime})|^{-1}\mu_{\beta,s}(k-k^{\prime
},t)dk^{\prime}\\
&  \leq\sup_{k^{\prime}\in\mathbb{R}}\left(  \mu_{\alpha,r}(k^{\prime
},t)\right)  \int_{\mathbb{R}}\frac{t^{\frac{s}{2}}}{|\tilde{k}|^{\frac{1}{2}%
}}\mu_{\beta,s}(\tilde{k},1)\frac{d\tilde{k}}{t^{s}}\\
&  \leq\frac{\mathrm{const.}}{t^{\frac{s}{2}}}\leq\frac{\mathrm{const.}%
}{t^{\frac{s}{2}}}\mu_{\alpha^{\prime},s^{\prime}}(k,t)~,
\end{align*}
where we have used the change of variables $k-k^{\prime}=\tilde{k}/t^{s}$. For
$k>t^{-s^{\prime}}$ and $s^{\prime}\leq s$ we have%
\begin{align*}
\left\vert (a\ast b)(k,t)\right\vert  &  \leq\int_{\mathbb{R}}\mu_{\alpha
,r}(k^{\prime},t)\frac{\mu_{\beta,s}(k-k^{\prime},t)}{|\kappa(k-k^{\prime}%
)|}dk^{\prime}\\
&  \leq\underset{:=I_{1}}{\underbrace{\int_{\mathbb{-\infty}}^{k/2}\mu
_{\alpha,r}(k^{\prime},t)\frac{\mu_{\beta,s}(k-k^{\prime},t)}{|\kappa
(k-k^{\prime})|}dk^{\prime}}}+\underset{:=I_{2}}{\underbrace{\int
_{k/2}^{\infty}\mu_{\alpha,r}(k^{\prime},t)\frac{\mu_{\beta,s}(k-k^{\prime
},t)}{|\kappa(k-k^{\prime})|}dk^{\prime}}}~.
\end{align*}
The integral $I_{2}$ is the same for (\ref{eq:convwithroot}) and
(\ref{eq:convwithrootgainbeta}),%
\begin{align*}
I_{2}  &  =\int_{k/2}^{\infty}\mu_{\alpha,r}(k^{\prime},t)\frac{1}%
{|\kappa(k-k^{\prime})|}\mu_{\beta,s}(k-k^{\prime},t)dk^{\prime}\\
&  \leq\mathrm{const.}~\mu_{\alpha,r}(k/2,t)\int_{\mathbb{R}}\frac{t^{\frac
{s}{2}}}{|\tilde{k}|^{\frac{1}{2}}}\mu_{\beta,s}(\tilde{k},1)\frac{d\tilde{k}%
}{t^{s}}\\
&  \leq\mathrm{const.}\frac{1}{t^{s/2}}\mu_{\alpha,r}(k,t)~,
\end{align*}
where again we have used the change of variables $k-k^{\prime}=\tilde{k}%
/t^{s}$. To compute the integral $I_{1}$ we use that%
\begin{align*}
\mu_{\alpha,s}\left(  k,t\right)   &  \leq\mu_{\alpha,s^{\prime}}\left(
k,t\right)  ~,\\
\mu_{\alpha,s}\left(  k,t\right)  \cdot\mu_{\beta,s}\left(  k,t\right)   &
\leq\mathrm{const.}\mu_{\alpha+\beta,s}\left(  k,t\right)  ~,
\end{align*}
and, for $k>t^{-s^{\prime}}$,%
\begin{align*}
\frac{1}{t^{\frac{s^{\prime}}{2}}}\frac{1}{|\kappa(k)|}  &  \leq
\frac{\mathrm{const.}}{t^{\frac{s^{\prime}}{2}}|k|^{1/2}}\leq\frac
{\mathrm{const.}}{2t^{\frac{s^{\prime}}{2}}|k|^{1/2}}\leq\frac{\mathrm{const.}%
}{1+\left(  t^{s^{\prime}}|k|\right)  ^{\frac{1}{2}}}\leq\mu_{\frac{1}%
{2},s^{\prime}}\left(  k,t\right)  ~,\\
\frac{1}{t^{s^{\prime}}}\frac{1}{|\kappa(k)|}  &  \leq\frac{\mathrm{const.}%
}{t^{s^{\prime}}\left(  |k|^{1/2}+|k|\right)  }\leq\frac{\mathrm{const.}%
}{t^{s^{\prime}/2}+|k|t^{s^{\prime}}}\leq\frac{\mathrm{const.}}%
{1+|k|t^{s^{\prime}}}\leq\mu_{1,s^{\prime}}\left(  k,t\right)  ~.
\end{align*}
To prove (\ref{eq:convwithroot}), we note that%
\begin{align*}
I_{1}^{(\ref{eq:convwithroot})}  &  \leq\int_{\mathbb{-\infty}}^{k/2}\frac
{\mu_{\alpha,r}(k^{\prime},t)}{|k^{\prime}|^{\frac{1}{2}}}\frac{|k^{\prime
}|^{\frac{1}{2}}}{|k-k^{\prime}|^{\frac{1}{2}}}\mu_{\beta,s}(k-k^{\prime
},t)dk^{\prime}\\
&  \leq\int_{\mathbb{-\infty}}^{k/2}\frac{\mu_{\alpha,r}(k^{\prime}%
,t)}{|k^{\prime}|^{\frac{1}{2}}}\frac{|k|^{\frac{1}{2}}}{|k-k^{\prime}%
|^{\frac{1}{2}}}\mu_{\beta,s}(k-k^{\prime},t)dk^{\prime}\\
&  +\int_{\mathbb{-\infty}}^{k/2}\frac{\mu_{\alpha,r}(k^{\prime}%
,t)}{|k^{\prime}|^{\frac{1}{2}}}\frac{|k-k^{\prime}|^{\frac{1}{2}}%
}{|k-k^{\prime}|^{\frac{1}{2}}}\mu_{\beta,s}(k-k^{\prime},t)dk^{\prime}\\
&  \leq\frac{|k|^{\frac{1}{2}}}{|k/2|^{\frac{1}{2}}}\mu_{\beta,s}%
(k/2,t)\int_{\mathbb{-\infty}}^{k/2}\frac{\mu_{\alpha,r}(k^{\prime}%
,t)}{|k^{\prime}|^{\frac{1}{2}}}dk^{\prime}\\
&  +\int_{\mathbb{-\infty}}^{k/2}\frac{\mu_{\alpha,r}(k^{\prime}%
,t)}{|k^{\prime}|^{\frac{1}{2}}}\frac{1}{|k-k^{\prime}|^{\frac{1}{2}}}%
\frac{\mathrm{const.}}{t^{s/2}}\mu_{\beta-\frac{1}{2},s}(k-k^{\prime
},t)dk^{\prime}\\
&  \leq\frac{\mathrm{const.}}{|k|^{\frac{1}{2}}}\frac{1}{t^{s/2}}\mu
_{\beta-\frac{1}{2},s}(k,t)\int_{\mathbb{-\infty}}^{k/2}\frac{\mu_{\alpha
,r}(k^{\prime},t)}{|k^{\prime}|^{\frac{1}{2}}}dk^{\prime}\\
&  \leq\mathrm{const.}\frac{t^{\frac{s^{\prime}}{2}}}{t^{\frac{s}{2}}}\frac
{1}{t^{\frac{s^{\prime}}{2}}|k|^{\frac{1}{2}}}\mu_{\beta-\frac{1}{2}%
,s}(k,t)\frac{1}{t^{\frac{r}{2}}}\\
&  \leq\mathrm{const.}\frac{t^{\frac{s^{\prime}}{2}}}{t^{\frac{s}{2}}}%
\mu_{\frac{1}{2},s^{\prime}}\left(  k,t\right)  \mu_{\beta-\frac{1}{2}%
,s}(k,t)\frac{1}{t^{\frac{r}{2}}}\leq\frac{\mathrm{const.}}{t^{\frac
{s+r-s^{\prime}}{2}}}\mu_{\beta,s^{\prime}}(k,t)~,
\end{align*}
where we have used the family of inequalities%
\begin{equation}
|k|^{\rho}\mu_{\alpha,r}\left(  k,t\right)  \leq\mathrm{const.}\frac
{1}{t^{\rho r}}\mu_{\alpha-p,r}\left(  k,t\right)  ~,\forall\rho>0~.
\label{eq:ksacrificealphafort}%
\end{equation}
Finally, to prove (\ref{eq:convwithrootgainbeta}), we note that%
\begin{align*}
I_{1}^{\left(  \ref{eq:convwithrootgainbeta}\right)  }  &  \leq\int
_{\mathbb{-\infty}}^{k/2}\mu_{\alpha,r}(k^{\prime},t)\frac{1}{|\kappa
(k-k^{\prime})|}\mu_{\beta,s}(k-k^{\prime},t)dk^{\prime}\\
&  \leq\frac{t^{cs^{\prime}}}{t^{cs^{\prime}}}\frac{1}{|\kappa(k/2)|}%
\mu_{\beta,s}(k/2,t)\int_{\mathbb{R}}\mu_{\alpha,r}(k^{\prime},t)dk^{\prime}\\
&  \leq\mathrm{const.}t^{cs^{\prime}}\mu_{c,s^{\prime}}(k,t\mu_{\beta
,s}(k,t)\int_{\mathbb{R}}\mu_{\alpha,r}(k^{\prime},t)dk^{\prime}\\
&  \leq\mathrm{const.}\frac{1}{t^{r-cs^{\prime}}}\mu_{\beta+c,s^{\prime}%
}(k,t)~.
\end{align*}
Collecting the bounds on the integrals $I_{1}^{(\ref{eq:convwithroot})}$,
$I_{1}^{\left(  \ref{eq:convwithrootgainbeta}\right)  }$ and $I_{2}$ proves
the claim in Proposition~\ref{prop:convwithroot}.
\end{proof}

\begin{corollary}
\label{corr:convsimplification}Let $\alpha>2$ and, for $i=1,2$, $p_{i}%
,q_{i}\geq0$. Let $f\in\mathcal{B}_{\alpha,p_{1},q_{1}}$ and $g\in
\mathcal{D}_{\alpha-1,p_{2},q_{2}}^{1}$. Let%
\begin{align*}
p  &  =\min\{p_{1}+p_{2}+\frac{1}{2},p_{1}+q_{2}+1,q_{1}+p_{2}+\frac{1}%
{2}\}~,\\
q  &  =\min\{q_{1}+q_{2}+1,q_{1}+p_{2}+\frac{1}{2}\}~.
\end{align*}
Then $f\ast g\in\mathcal{B}_{\alpha,p,q}$, and there exists a constant $C$,
depending only on $\alpha$, such that%
\[
\Vert f\ast g;\mathcal{B}_{\alpha,p,q}\Vert\leq C~\Vert f;\mathcal{B}%
_{\alpha,p_{1},q_{1}}\Vert\cdot\Vert g;\mathcal{D}_{\alpha-1,p_{2},q_{2}}%
^{1}\Vert~.
\]

\end{corollary}

\begin{proof}
We consider the three cases $c\in\{0,\frac{1}{2},1\}$. Let $\tilde{g}$ be a
function in $\mathcal{B}_{\tilde{\alpha},\tilde{p},\tilde{q}}^{1}$, with
$\tilde{\alpha}=\alpha-c$, $\tilde{p},\tilde{q}\geq0$. The convolution product
$f\ast\tilde{g}$ is in each case bounded by a function in $\mathcal{B}%
_{\alpha,p,q}^{1}$ with $p$ and $q$ given by:

\begin{itemize}
\item if $c=0$, $p=\min\{p_{1}+\tilde{p}+\frac{1}{2},p_{1}+\tilde{q}%
+1,q_{1}+\tilde{p}+\frac{1}{2}\}$~, $q=\min\{q_{1}+\tilde{q}+1,q_{1}+\tilde
{p}+\frac{1}{2}\}$~,

\item if $c=\frac{1}{2}$, $p=\min\{p_{1}+\tilde{p}+\frac{1}{2},p_{1}+\tilde
{q}+\frac{1}{2},q_{1}+\tilde{p}+\frac{1}{2}\}~$, $q=\min\{q_{1}+\tilde
{q}+1,q_{1}+\tilde{p}+\frac{1}{2}\}$~,

\item if $c=1$, $p=\min\{p_{1}+\tilde{p}+0,p_{1}+\tilde{q}+0,q_{1}+\tilde
{p}+\frac{1}{2}\}$~, $q=\min\{q_{1}+\tilde{q}+0,q_{1}+\tilde{p}+\frac{1}{2}\}$~.
\end{itemize}

These are consequences of Proposition~\ref{prop:convwithroot}. Using equation
(\ref{eq:convwithroot}) for the first case and equation
(\ref{eq:convwithrootgainbeta}) for the following two cases, and choosing
$s^{\prime}=1$ to bound the term $\frac{1}{t^{p_{1}}}\bar{\mu}_{\alpha}%
\ast\frac{1}{t^{\tilde{q}}}\tilde{\mu}_{\tilde{\alpha}}$. It is now clear that
for a function in $\mathcal{D}_{\alpha-1,p_{2},q_{2}}^{1}$, the terms that
yield the lowest $p$ and $q$ are covered by the $c=0$ case above, because what
is lost in the bounds on convolution due to lower $\tilde{\alpha}$ is gained
through higher values of $\tilde{p}$ and $\tilde{q}$ by definition of the
space $\mathcal{D}_{\alpha-1,p_{2},q_{2}}^{1}$. This corollary allows to
streamline notations and shorten calculations throughout the paper.
\end{proof}

\section{\label{sec:bounds}Bounds on $\mathbf{\hat{d}}$}

We present some elementary inequalities and expressions used throughout this
section. Throughout the calculations we will use without further mention, that
for all $z\in\mathbb{C}$ with $\operatorname{Re}(z)\leq0$ and $N\in
\mathbb{N}_{0}$,
\[
\left\vert \frac{e^{z}-\sum_{n=0}^{N}\frac{1}{n!}z^{n}}{z^{N+1}}\right\vert
\leq\mathrm{const.}~,
\]
and for all $z\in\mathbb{C}$ with $\operatorname{Re}(z)>0$%
\[
\left\vert \frac{e^{z}-\sum_{n=0}^{N}\frac{1}{n!}z^{n}}{z^{N+1}}\right\vert
\leq\mathrm{const.}e^{\operatorname{Re}(z)}~.
\]
We also have that%
\[
\partial_{k}\kappa=\frac{2k-i}{2\kappa}~.
\]
By definition of the norm on $\mathcal{D}_{\alpha,p,q}^{1}$ we must bound
$\kappa\partial_{k}\hat{\omega}$. We thus bound all the terms $\kappa
\partial_{k}\hat{\omega}_{l,n,m}$, with $l=1,2,3$, $n=1,2,3$ and $m=0,1$ (see
definitions (\ref{eq:dkomegaInt}), (\ref{eq:dkomegaIntComp1}%
)-(\ref{eq:dkomegaIntComp3}) and (\ref{eq:f10})-(\ref{eq:dkf31})). This
requires a good deal of book-keeping to track what happens to $\alpha$, $p$,
and $q$. Some of it may be spared when one realizes that all losses in
$\alpha$ occur when applying (\ref{eq:ksacrificealphafort}) where there are
explicit factors $|k|^{c}$ with $c=\{\frac{1}{2},1\}$, which automatically
brings forth a structure satisfying the conditions of
Corollary~\ref{corr:convsimplification}. This allows us to show that each
component $\partial_{k}\hat{\omega}_{l,n,m}$ is an element of a $\mathcal{D}%
_{\alpha-1,p,q}^{1}$.

From (\ref{eq:Qispace}) we obtain, for $i=0,1$,%
\[
\left\vert \hat{Q}_{i}\left(  k,s\right)  \right\vert \leq\left\Vert \hat
{Q}_{i};\mathcal{B}_{\alpha,\frac{7}{2},\frac{5}{2}}\right\Vert \left(
\frac{1}{s^{\frac{7}{2}}}\bar{\mu}_{\alpha}+\frac{1}{s^{\frac{5}{2}}}%
\tilde{\mu}_{\alpha}\right)  ~,
\]
which we will use throughout without further mention. We also make use of
equations (\ref{eq:sqrtKleqkappa}) and (\ref{eq:ksacrificealphafort}) without
explicit mention throughout these proofs.

\vspace{0in}

The bounds for the terms $n=2$ take advantage of the fact that, for $1\leq
t<2$,%
\[
\bar{\mu}_{\alpha}(k,t)\leq\mathrm{const.}~\tilde{\mu}_{\alpha}(k,t)\leq
\mathrm{const.}%
\]
and, for $t\geq2$ and $\alpha^{\prime}>0$,%
\[
e^{\Lambda_{-}(t-1)}\mu_{\alpha,r}(k,t)\leq\mathrm{const.}~e^{\Lambda
_{-}(t-1)}\leq\mathrm{const.}~\tilde{\mu}_{\alpha^{\prime}}(k,t)~,
\]
so that the inequality%
\begin{equation}
e^{\Lambda_{-}(t-1)}\mu_{\alpha,r}(k,t)\leq\mathrm{const.}~\tilde{\mu}%
_{\alpha}(k,t) \label{eq:mutomutilde}%
\end{equation}
holds for all $t$ and $\alpha>0$.

\subsection{\label{sec:boundsd1}Bounds on $\hat{d}_{1}$}

To show that $\hat{d}_{1}=\sum_{m=0}^{1}\sum_{n=1}^{3}\partial_{k}\hat{\omega
}_{1,n,m}$ is in $\mathcal{D}_{\alpha-1,\frac{3}{2},0}^{1}$, which constitutes
the first part of Proposition~\ref{prop:d1&d2space}, we first need to recall a
proposition proved in \cite{Hillairet.Wittwer-Existenceofstationary2009}.

\begin{proposition}
\label{prop:fnmbounds}Let $f_{n,m}$ be as given in Section$~$%
\ref{sec:integraleq}. Then we have the bounds%
\begin{align}
\left\vert f_{1,0}(k,\sigma)\right\vert  &  \leq\mathrm{const.}~e^{|\Lambda
_{-}|\sigma}\min\{|\Lambda_{-}|,|\Lambda_{-}|^{3}\sigma^{2}%
\}~,\label{eq:bndf10}\\
\left\vert f_{2,0}(k,\sigma)\right\vert  &  \leq\mathrm{const.}~(\left\vert
k\right\vert +\left\vert k\right\vert ^{1/2})e^{-\left\vert k\right\vert
\sigma}~,\label{eq:bndf20}\\
\left\vert f_{3,0}(k,\sigma)\right\vert  &  \leq\mathrm{const.}~e^{\Lambda
_{-}\sigma}\min\{1,|\Lambda_{-}|^{2}\}~,\label{eq:bndf30}\\
\left\vert f_{1,1}(k,\sigma)\right\vert  &  \leq\mathrm{const.}~(1+|\Lambda
_{-}|)e^{|\Lambda_{-}|\sigma}\min\{1,|\Lambda_{-}|\sigma\}~, \label{eq:bndf11}%
\\
\left\vert f_{2,1}(k,\sigma)\right\vert  &  \leq\mathrm{const.}~\left(
1+|k|\right)  e^{-\left\vert k\right\vert \sigma}~,\label{eq:bndf21}\\
\left\vert f_{3,1}(k,\sigma)\right\vert  &  \leq\mathrm{const.}~e^{\Lambda
_{-}\sigma}\min\{1,|\Lambda_{-}|\}~, \label{eq:bndf31}%
\end{align}
uniformly in $\sigma\geq0$ and $k\in\mathbb{R}_{0}$.
\end{proposition}

\bigskip

We then note that%
\begin{align*}
\left\vert \kappa\partial_{k}\check{K}_{n}(k,\tau)\right\vert  &  =\left\vert
\frac{1}{2}\tau\kappa\frac{2k-i}{2\kappa}e^{-\kappa\tau}\right\vert
\leq\mathrm{const.}~\tau(1+|k|)e^{\Lambda_{-}\tau}~,\text{ for }n=1,2~,\\
\left\vert \kappa\partial_{k}\check{K}_{3}(k,\tau)\right\vert  &  =\left\vert
\frac{1}{2}\tau\kappa\frac{2k-i}{2\kappa}(e^{\kappa\tau}+e^{-\kappa\tau
})\right\vert \leq\mathrm{const.}~\tau(1+|k|)(e^{|\Lambda_{-}|\tau}%
+e^{\Lambda_{-}\tau})~.
\end{align*}
The bound on the function $\kappa\partial_{k}\hat{\omega}_{1,1,0}$ uses
(\ref{eq:bndf10}) and Propositions~\ref{prop:sgL1} and \ref{prop:sgL2},
leading to
\begin{align*}
&  \left\vert \kappa\partial_{k}\hat{\omega}_{1,1,0}\right\vert =\left\vert
\kappa\frac{1}{2}\partial_{k}e^{-\kappa\tau}\int_{1}^{t}\check{f}_{1,0}\left(
k,\sigma\right)  \hat{Q}_{0}\left(  k,s\right)  ds\right\vert \\
&  \leq\mathrm{const.}~\tau(1+|k|)e^{\Lambda_{-}\tau}\int_{1}^{t}%
e^{|\Lambda_{-}|\sigma}\min\{|\Lambda_{-}|,|\Lambda_{-}|^{3}\sigma
^{2}\}\left(  \frac{1}{s^{\frac{7}{2}}}\bar{\mu}_{\alpha}+\frac{1}{s^{\frac
{5}{2}}}\tilde{\mu}_{\alpha}\right)  ds\\
&  \leq\mathrm{const.}~\tau(1+|k|)e^{\Lambda_{-}\tau}\int_{1}^{\frac{t+1}{2}%
}e^{|\Lambda_{-}|\sigma}|\Lambda_{-}|^{3}\sigma^{2}\left(  \frac{1}%
{s^{\frac{7}{2}}}\bar{\mu}_{\alpha}+\frac{1}{s^{\frac{5}{2}}}\tilde{\mu
}_{\alpha}\right)  ds\\
&  +\mathrm{const.}~\tau(1+|k|)e^{\Lambda_{-}\tau}\int_{\frac{t+1}{2}}%
^{t}e^{|\Lambda_{-}|\sigma}|\Lambda_{-}|\left(  \frac{1}{s^{\frac{7}{2}}}%
\bar{\mu}_{\alpha}+\frac{1}{s^{\frac{5}{2}}}\tilde{\mu}_{\alpha}\right)  ds\\
&  \leq\mathrm{const.}~(1+|k|)\left(  \frac{1}{t^{\frac{5}{2}}}\bar{\mu
}_{\alpha}+\frac{1}{t^{\frac{3}{2}}}\tilde{\mu}_{\alpha}\right)  ~,
\end{align*}
which shows that $\kappa\partial_{k}\hat{\omega}_{1,1,0}\in\mathcal{D}%
_{\alpha-1,\frac{5}{2},\frac{3}{2}}^{1}$.

The bound on the function $\kappa\partial_{k}\hat{\omega}_{1,2,0}$ uses
(\ref{eq:bndf20}), Proposition~\ref{prop:sgk3} and (\ref{eq:mutomutilde}),
leading to%
\begin{align*}
\left\vert \kappa\partial_{k}\hat{\omega}_{1,2,0}\right\vert  &  =\left\vert
\kappa\frac{1}{2}\partial_{k}e^{-\kappa\tau}\int_{t}^{\infty}\check{f}%
_{2,0}(k,s-1)\hat{Q}_{0}(k,s)ds\right\vert \\
&  \leq\mathrm{const.}~\tau(1+|k|)e^{\Lambda_{-}\tau}e^{|k|\tau}\int
_{t}^{\infty}(|k|^{\frac{1}{2}}+|k|)e^{-|k|\sigma}\left(  \frac{1}{s^{\frac
{7}{2}}}\bar{\mu}_{\alpha}+\frac{1}{s^{\frac{5}{2}}}\tilde{\mu}_{\alpha
}\right)  ds\\
&  \leq\mathrm{const.}~(1+|k|)e^{\Lambda_{-}\tau}\left(  \frac{1}{t^{2}}%
\bar{\mu}_{\alpha}+\frac{1}{t^{1}}\tilde{\mu}_{\alpha}\right)  \leq
\mathrm{const.}~(1+|k|)\frac{1}{t^{1}}\tilde{\mu}_{\alpha}~,
\end{align*}
which shows that $\kappa\partial_{k}\hat{\omega}_{1,2,0}\in\mathcal{D}%
_{\alpha-1,\infty,1}^{1}$.

The bound on the function $\kappa\partial_{k}\hat{\omega}_{1,3,0}$ uses
(\ref{eq:bndf30}) and Proposition~\ref{prop:sgL3}, leading to%
\begin{align*}
\left\vert \kappa\partial_{k}\hat{\omega}_{1,3,0}\right\vert  &  =\left\vert
\frac{1}{2}\kappa\partial_{k}(e^{\kappa\tau}-e^{-\kappa\tau})\int_{t}^{\infty
}\check{f}_{3,0}(k,s-1)\hat{Q}_{0}(k,s)ds\right\vert \\
&  \leq\mathrm{const.}~\tau(1+|k|)(e^{|\Lambda_{-}|\tau}+e^{\Lambda_{-}\tau
})\int_{t}^{\infty}\min\{1,|\Lambda_{-}|\}e^{\Lambda_{-}\sigma}\left(
\frac{1}{s^{\frac{7}{2}}}\bar{\mu}_{\alpha}+\frac{1}{s^{\frac{5}{2}}}%
\tilde{\mu}_{\alpha}\right)  ds\\
&  \leq\mathrm{const.}~\tau e^{|\Lambda_{-}|\tau}\int_{t}^{\infty}%
(1+|\Lambda_{-}|)\min\{1,|\Lambda_{-}|\}e^{\Lambda_{-}\sigma}\left(  \frac
{1}{s^{\frac{7}{2}}}\bar{\mu}_{\alpha}+\frac{1}{s^{\frac{5}{2}}}\tilde{\mu
}_{\alpha}\right)  ds\\
&  \leq\mathrm{const.}\left(  \frac{1}{t^{\frac{5}{2}}}\bar{\mu}_{\alpha
}+\frac{1}{t^{\frac{3}{2}}}\tilde{\mu}_{\alpha}\right)  ~,
\end{align*}
which shows that $\kappa\partial_{k}\hat{\omega}_{1,3,0}\in\mathcal{D}%
_{\alpha-1,\frac{5}{2},\frac{3}{2}}^{1}$.

The bound on the function $\kappa\partial_{k}\hat{\omega}_{1,1,1}$ uses
(\ref{eq:bndf11}) and Propositions~\ref{prop:sgL1} and \ref{prop:sgL2},
leading to%
\begin{align*}
\left\vert \kappa\partial_{k}\hat{\omega}_{1,1,1}\right\vert  &  =\left\vert
\kappa\frac{1}{2}\partial_{k}e^{-\kappa\tau}\int_{1}^{t}\check{f}%
_{1,1}(k,s-1)\hat{Q}_{1}(k,s)ds\right\vert \\
&  \leq\mathrm{const.}~\tau(1+|k|)e^{\Lambda_{-}\tau}\int_{1}^{t}%
(1+|\Lambda_{-}|)e^{|\Lambda_{-}|\sigma}\min\{1,|\Lambda_{-}|\sigma\}\left(
\frac{1}{s^{\frac{7}{2}}}\bar{\mu}_{\alpha}+\frac{1}{s^{\frac{5}{2}}}%
\tilde{\mu}_{\alpha}\right)  ds\\
&  \leq\mathrm{const.}~\tau(1+|k|)e^{\Lambda_{-}\tau}\int_{1}^{t}%
e^{|\Lambda_{-}|\sigma}|\Lambda_{-}|\sigma\left(  \frac{1}{s^{\frac{7}{2}}%
}\bar{\mu}_{\alpha}+\frac{1}{s^{\frac{5}{2}}}\tilde{\mu}_{\alpha}\right)  ds\\
&  +\mathrm{const.}~\tau(1+|k|)e^{\Lambda_{-}\tau}\int_{1}^{t}|\Lambda
_{-}|e^{|\Lambda_{-}|\sigma}\min\{1,|\Lambda_{-}|\sigma\}\left(  \frac
{1}{s^{\frac{7}{2}}}\bar{\mu}_{\alpha}+\frac{1}{s^{\frac{5}{2}}}\tilde{\mu
}_{\alpha}\right)  ds\\
&  \leq\mathrm{const.}~(1+|k|)\left(  \tilde{\mu}_{\alpha}+\frac{1}%
{t^{\frac{3}{2}}}\bar{\mu}_{\alpha}+\frac{1}{t^{\frac{1}{2}}}\tilde{\mu
}_{\alpha}\right)  ~,
\end{align*}
which shows that $\kappa\partial_{k}\hat{\omega}_{1,1,1}\in\mathcal{D}%
_{\alpha-1,\frac{3}{2},0}^{1}$.

The bound on the function $\kappa\partial_{k}\hat{\omega}_{1,2,1}$ uses
(\ref{eq:bndf21}), Proposition~\ref{prop:sgk3} and (\ref{eq:mutomutilde}),
leading to%
\begin{align*}
\left\vert \kappa\partial_{k}\hat{\omega}_{1,2,1}\right\vert  &  =\left\vert
\frac{1}{2}\kappa\partial_{k}e^{-\kappa\tau}\int_{t}^{\infty}\check{f}%
_{2,1}(k,s-1)\hat{Q}_{1}(k,s)ds\right\vert \\
&  \leq\mathrm{const.}~(1+|k|)\tau e^{\Lambda_{-}\tau}\int_{t}^{\infty
}(1+|k|)e^{-|k|\sigma}\left(  \frac{1}{s^{\frac{7}{2}}}\bar{\mu}_{\alpha
}+\frac{1}{s^{\frac{5}{2}}}\tilde{\mu}_{\alpha}\right)  ds\\
&  \leq\mathrm{const.}~(1+|k|)e^{\Lambda_{-}\tau}\left(  \frac{1}{t^{\frac
{3}{2}}}\bar{\mu}_{\alpha}+\frac{1}{t^{\frac{1}{2}}}\tilde{\mu}_{\alpha
}\right)  \leq\mathrm{const.}~(1+|k|)\frac{1}{t^{\frac{1}{2}}}\tilde{\mu
}_{\alpha}~,
\end{align*}
which shows that $\kappa\partial_{k}\hat{\omega}_{1,2,1}\in\mathcal{D}%
_{\alpha-1,\infty,\frac{1}{2}}^{1}$.

The bound on the function $\kappa\partial_{k}\hat{\omega}_{1,3,1}$ uses
(\ref{eq:bndf31}) and Proposition~\ref{prop:sgL3}, leading to%
\begin{align*}
\left\vert \kappa\partial_{k}\hat{\omega}_{1,3,1}\right\vert  &  =\left\vert
\kappa\frac{1}{2}\partial_{k}(e^{\kappa\tau}-e^{-\kappa\tau})\int_{t}^{\infty
}\check{f}_{3,1}\hat{Q}_{1}(k,s)ds\right\vert \\
&  \leq\mathrm{const.}~\tau(1+|k|)(e^{|\Lambda_{-}|\tau}+e^{\Lambda_{-}\tau
})\int_{t}^{\infty}e^{\Lambda_{-}\sigma}\left(  \frac{1}{s^{\frac{7}{2}}}%
\bar{\mu}_{\alpha}+\frac{1}{s^{\frac{5}{2}}}\tilde{\mu}_{\alpha}\right)  ds\\
&  \leq\mathrm{const.}~\tau(e^{|\Lambda_{-}|\tau}+e^{\Lambda_{-}\tau})\int
_{t}^{\infty}(1+|\Lambda_{-}|)e^{\Lambda_{-}\sigma}\left(  \frac{1}%
{s^{\frac{7}{2}}}\bar{\mu}_{\alpha}+\frac{1}{s^{\frac{5}{2}}}\tilde{\mu
}_{\alpha}\right)  ds\\
&  \leq\mathrm{const.}\left(  \frac{1}{t^{\frac{3}{2}}}\bar{\mu}_{\alpha
}(k,t)+\frac{1}{t^{\frac{1}{2}}}\tilde{\mu}_{\alpha}(k,t)\right)  ~,
\end{align*}
which shows that $\kappa\partial_{k}\hat{\omega}_{1,3,0}\in\mathcal{D}%
_{\alpha-1,\frac{3}{2},\frac{1}{2}}^{1}$.

Collecting the bounds we find that $\hat{d}_{1}\in\mathcal{D}_{\alpha
-1,\frac{3}{2},0}^{1}$~, which completes the first part of the proof of
Proposition~\ref{prop:d1&d2space}.

\subsection{\label{sec:boundsd2}Bounds on $\hat{d}_{2}$}

To show that $\hat{d}_{2}=\sum_{m=0}^{1}\sum_{n=1}^{3}\partial_{k}\hat{\omega
}_{2,n,m}$ is in $\mathcal{D}_{\alpha-1,\frac{3}{2},0}^{1}$, which constitutes
the second part of Proposition~\ref{prop:d1&d2space}, we first need to show
bounds on the functions $\partial_{k}\check{f}_{n,m}$.

\begin{proposition}
\label{prop:kappadkfnmbounds}Let $\partial_{k}\check{f}_{n,m}$ be as given in
Section$~$\ref{sec:integraleq}. Then we have the bounds%
\begin{align}
\left\vert \kappa\partial_{k}\check{f}_{1,0}(k,\sigma)\right\vert  &
\leq\mathrm{const.}~\min\{(1+|\Lambda_{-}|\sigma),(s+|\Lambda_{-}%
|)|\Lambda_{-}|^{2}\sigma\}e^{|\Lambda_{-}|\sigma}~,\label{eq:bnddkf10}\\
\left\vert \kappa\partial_{k}\check{f}_{2,0}(k,\sigma)\right\vert  &
\leq\mathrm{const.}~(|k|^{\frac{1}{2}}+|k|^{2})\sigma e^{-|k|\sigma
}~,\label{eq:bnddkf20}\\
\left\vert \kappa\partial_{k}\check{f}_{3,0}(k,\sigma)\right\vert  &
\leq\mathrm{const.}~(1+|\Lambda_{-}|\sigma)e^{\Lambda_{-}\sigma}%
~,\label{eq:bnddkf30}\\
\left\vert \kappa\partial_{k}\check{f}_{1,1}(k,\sigma)\right\vert  &
\leq\mathrm{const.}~(1+|\Lambda_{-}|^{2})\sigma e^{|\Lambda_{-}|\sigma
}~,\label{eq:bnddkf11}\\
\left\vert \kappa\partial_{k}\check{f}_{2,1}(k,\sigma)\right\vert  &
\leq\mathrm{const.}~(1+|k|^{2})\sigma e^{-|k|\sigma}~,\label{eq:bnddkf21}\\
\left\vert \kappa\partial_{k}\check{f}_{3,1}(k,\sigma)\right\vert  &
\leq\mathrm{const.}~(1+|\Lambda_{-}|)\sigma e^{\Lambda_{-}\sigma}~,
\label{eq:bnddkf31}%
\end{align}
uniformly in $\sigma\geq0$ and $k\in\mathbb{R}_{0}$.
\end{proposition}

\begin{proof}
We multiply (\ref{eq:dkf10})-(\ref{eq:dkf31}) by $\kappa$ and bound the
products. The function $\kappa\partial_{k}\check{f}_{1,0}$ is bounded in two
ways. We have a straightforward bound%
\[
\left\vert \kappa\partial_{k}\check{f}_{1,0}(k,\sigma)\right\vert
\leq\mathrm{const.}~(1+|\Lambda_{-}|\sigma)e^{|\Lambda_{-}|\sigma}~.
\]
Since leading terms cancel, we get%
\begin{align*}
\left\vert \kappa\partial_{k}\check{f}_{1,0}(k,\sigma)\right\vert  &
\leq\left\vert \frac{i}{2}\left(  e^{\kappa\sigma}+e^{-\kappa\sigma
}-2e^{-|k|\sigma}\right)  -\frac{ik^{2}}{2\kappa^{2}}\left(  e^{\kappa\sigma
}-e^{-\kappa\sigma}\right)  +2\frac{k^{2}+|k|\kappa}{k}(e^{-|k|\sigma
}-e^{-\kappa\sigma})\right\vert \\
&  +\left\vert i\frac{k^{2}+\kappa^{2}}{2\kappa}\left(  e^{\kappa\sigma
}-e^{-\kappa\sigma}\right)  \sigma+\frac{k^{2}+|k|\kappa}{k}\frac{k^{2}%
+\kappa^{2}}{\kappa}e^{-\kappa\sigma}\sigma-2\kappa\frac{k^{2}+|k|\kappa}%
{k}e^{-|k|\sigma}\sigma\right\vert \\
&  \leq\mathrm{const.}~|(e^{\kappa\sigma}-1-\kappa\sigma)+(e^{-\kappa\sigma
}-1+\kappa\sigma)-2(e^{-|k|\sigma}-1)|\\
&  +\mathrm{const.}\left\vert \frac{k^{2}}{\kappa^{2}}\left(  (e^{\kappa
\sigma}-1)-(e^{-\kappa\sigma}-1)\right)  \right\vert +\mathrm{const.}%
~|\Lambda_{-}||(e^{-|k|\sigma}-1)-(e^{-\kappa\sigma}-1)|\\
&  +\mathrm{const.}\left\vert \frac{k^{2}+\kappa^{2}}{2\kappa}\left(
(e^{\kappa\sigma}-1)-(e^{-\kappa\sigma}-1)\right)  \sigma\right\vert \\
&  +\mathrm{const.}\left\vert \frac{k^{2}+|k|\kappa}{k}\frac{k^{2}+\kappa^{2}%
}{\kappa}e^{-\kappa\sigma}\sigma\right\vert +\mathrm{const.}\left\vert
\kappa\frac{k^{2}+|k|\kappa}{k}e^{-|k|\sigma}\sigma\right\vert \\
&  \leq\mathrm{const.}~(|\Lambda_{-}|^{2}\sigma^{2}+|k|\sigma)e^{|\Lambda
_{-}|\sigma}+\mathrm{const.}~|\Lambda_{-}|^{3}\sigma e^{|\Lambda_{-}|\sigma
}+\mathrm{const.}~|\Lambda_{-}|^{2}\sigma e^{|\Lambda_{-}|\sigma}\\
&  +\mathrm{const.}~|\Lambda_{-}|^{2}\sigma^{2}e^{|\Lambda_{-}|\sigma
}+\mathrm{const.}~|\Lambda_{-}|^{2}\sigma e^{|\Lambda_{-}|\sigma
}+\mathrm{const.}~|\Lambda_{-}|^{2}\sigma e^{|\Lambda_{-}|\sigma}\\
&  \leq\mathrm{const.}~(s+|\Lambda_{-}|)|\Lambda_{-}|^{2}\sigma e^{|\Lambda
_{-}|\sigma}~.
\end{align*}
Then we have%
\[
\left\vert \kappa\partial_{k}\check{f}_{1,0}(k,\sigma)\right\vert
\leq\mathrm{const.}~\min\{(1+|\Lambda_{-}|\sigma),(s+|\Lambda_{-}%
|)|\Lambda_{-}|^{2}\sigma\}e^{|\Lambda_{-}|\sigma}~,
\]
which proves (\ref{eq:bnddkf10}).

To bound $\kappa\partial_{k}\check{f}_{2,0}(k,\sigma)$ we use that, since
$|k|\leq\operatorname{Re}(\kappa)$ for all $k$,%
\begin{align}
\left\vert e^{-|k|\sigma}-e^{-\kappa\sigma}\right\vert  &  \leq\mathrm{const.}%
~e^{-|k|\sigma}\left\vert 1-e^{(|k|-\kappa)\sigma}\right\vert \nonumber\\
&  \leq\mathrm{const.}~e^{-|k|\sigma}\left\vert |k|-\kappa\right\vert
\sigma\nonumber\\
&  \leq\mathrm{const.}~(|k|^{\frac{1}{2}}+|k|)\sigma e^{-|k|\sigma}~,
\label{eq:exp(abs(k)-kappa)}%
\end{align}
such that%
\begin{align*}
\left\vert \kappa\partial_{k}\check{f}_{2,0}(k,\sigma)\right\vert  &
\leq\left\vert \frac{(|k|+\kappa)^{2}}{k}\left(  e^{-|k|\sigma}-e^{-\kappa
\sigma}\right)  -2\frac{\kappa+|k|}{k}\left(  |k|\kappa e^{-|k|\sigma}%
-\frac{k^{2}+\kappa^{2}}{2}e^{-\kappa\sigma}\right)  \sigma\right\vert \\
&  \leq\mathrm{const.}~(1+|k|)(|k|^{\frac{1}{2}}+|k|)e^{-|k|\sigma}%
\sigma+\mathrm{const.}~(|k|+|k|^{2})e^{-|k|\sigma}\sigma\\
&  +\mathrm{const.}~(|k|^{\frac{1}{2}}+|k|^{2})e^{-|k|\sigma}\sigma\\
&  \leq\mathrm{const.}~(|k|^{\frac{1}{2}}+|k|^{2})\sigma e^{-|k|\sigma}~,
\end{align*}
which gives (\ref{eq:bnddkf20}).

To bound $\kappa\partial_{k}\check{f}_{3,0}\left(  k,\sigma\right)  $ we have
the straightforward bound%
\begin{align*}
\left\vert \kappa\partial_{k}\check{f}_{3,0}\left(  k,\sigma\right)
\right\vert  &  \leq\left\vert \kappa\frac{k}{2\kappa^{3}}e^{-\kappa\sigma
}\right\vert +\left\vert \kappa\frac{k^{2}+\kappa^{2}}{2\kappa^{2}}\sigma
e^{-\kappa\sigma}\right\vert \\
&  \leq\mathrm{const.}~(1+|\Lambda_{-}|)\sigma e^{\Lambda_{-}\sigma}~,
\end{align*}
which yields (\ref{eq:bnddkf30}).

To bound $\kappa\partial_{k}\check{f}_{1,1}\left(  k,\sigma\right)  $ we have%
\begin{align*}
\left\vert \kappa\partial_{k}\check{f}_{1,1}(k,\sigma)\right\vert  &
\leq\left\vert i\frac{\left(  |k|+\kappa\right)  ^{2}}{|k|}(e^{-|k|\sigma
}-e^{-\kappa\sigma})\right\vert +\left\vert \frac{k^{2}+\kappa^{2}}{2k}\left(
e^{\kappa\sigma}+e^{-\kappa\sigma}\right)  \sigma\right\vert \\
&  +\left\vert 2i\frac{k^{2}+|k|\kappa}{k^{2}}\left(  \frac{k^{2}+\kappa^{2}%
}{2}e^{-\kappa\sigma}-|k|\kappa e^{-|k|\sigma}\right)  \sigma\right\vert \\
&  \leq\mathrm{const.}~(1+|k|)(|k|+|\Lambda_{-}|)\sigma+\mathrm{const.}%
~(1+|k|)\sigma e^{|\Lambda_{-}%
\vert
\sigma}\\
&  +\mathrm{const.}~|\Lambda_{-}|((1+|k|)+|\Lambda_{-}|)\sigma\leq
\mathrm{const.}~(1+|\Lambda_{-}|^{2})\sigma e^{|\Lambda_{-}|\sigma}~,
\end{align*}
and thus we have (\ref{eq:bnddkf11}).

To bound $\kappa\partial_{k}\check{f}_{2,1}\left(  k,\sigma\right)  $ we use
(\ref{eq:exp(abs(k)-kappa)}) to bound%
\begin{align*}
\left\vert \kappa\partial_{k}\check{f}_{2,1}\left(  k,\sigma\right)
\right\vert  &  \leq\left\vert i\frac{\left(  |k|+\kappa\right)  ^{2}}%
{|k|}(e^{-\kappa\sigma}-e^{-|k|\sigma})\right\vert +\left\vert i(|k|+\kappa
)\kappa\frac{k^{2}+\kappa^{2}}{k^{2}}e^{-\kappa\sigma}\sigma\right\vert \\
&  +|2i\kappa\left(  |k|+\kappa\right)  e^{-|k|\sigma}\sigma|\\
&  \leq\mathrm{const.}~(1+|k|)(|k|^{\frac{1}{2}}+|k|)\sigma e^{-|k|\sigma
}+\mathrm{const.}~(|k|+|k|^{2})(1+|k|^{-1})e^{-|k|\sigma}\sigma\\
&  +\mathrm{const.}~(|k|+|k|^{2})e^{-|k|\sigma}\sigma\leq\mathrm{const.}%
~(1+|k|^{2})\sigma e^{-|k|\sigma}%
\end{align*}
which leads to (\ref{eq:bnddkf21}).

Finally, To bound $\kappa\partial_{k}\check{f}_{3,1}\left(  k,\sigma\right)  $
we have the straightforward bound%
\[
\left\vert \kappa\partial_{k}\check{f}_{3,1}\left(  k,\sigma\right)
\right\vert \leq\left\vert \frac{k^{2}+\kappa^{2}}{2k}e^{-\kappa\sigma}%
\sigma\right\vert \leq\mathrm{const.}~(1+|\Lambda_{-}|)\sigma e^{\Lambda
_{-}\sigma}~,
\]
and therefore we have (\ref{eq:bnddkf31}). This completes the proof of
Proposition~\ref{prop:kappadkfnmbounds}.
\end{proof}

\bigskip

We may now bound $\hat{d}_{2}$. The bound on the function $\kappa\partial
_{k}\hat{\omega}_{2,1,0}$ uses (\ref{eq:bnddkf10}) and
Propositions~\ref{prop:sgL1} and \ref{prop:sgL2}, leading to%
\begin{align*}
\left\vert \kappa\partial_{k}\hat{\omega}_{2,1,0}\right\vert  &  =\left\vert
\frac{1}{2}e^{-\kappa\tau}\int_{1}^{t}\kappa\partial_{k}\check{f}_{1,0}\left(
k,\sigma\right)  \hat{Q}_{0}\left(  k,s\right)  ds\right\vert \\
&  \leq\mathrm{const.}~e^{\Lambda_{-}\tau}\int_{1}^{t}\min\{(1+|\Lambda
_{-}|\sigma),(s+|\Lambda_{-}|)|\Lambda_{-}|^{2}\sigma\}e^{|\Lambda_{-}|\sigma
}\left(  \frac{1}{s^{\frac{7}{2}}}\bar{\mu}_{\alpha}+\frac{1}{s^{\frac{5}{2}}%
}\tilde{\mu}_{\alpha}\right)  ds\\
&  \leq\mathrm{const.}~e^{\Lambda_{-}\tau}\int_{1}^{\frac{t+1}{2}}%
(s+|\Lambda_{-}|)|\Lambda_{-}|^{2}\sigma e^{|\Lambda_{-}|\sigma}\left(
\frac{1}{s^{\frac{7}{2}}}\bar{\mu}_{\alpha}+\frac{1}{s^{\frac{5}{2}}}%
\tilde{\mu}_{\alpha}\right)  ds\\
&  +\mathrm{const.}~e^{\Lambda_{-}\tau}\int_{\frac{t+1}{2}}^{t}(1+|\Lambda
_{-}|\sigma)e^{|\Lambda_{-}|\sigma}\left(  \frac{1}{s^{\frac{7}{2}}}\bar{\mu
}_{\alpha}+\frac{1}{s^{\frac{5}{2}}}\tilde{\mu}_{\alpha}\right)  ds\\
&  \leq\mathrm{const.}~(1+|\Lambda_{-}|)\frac{1}{t^{\frac{3}{2}}}\tilde{\mu
}_{\alpha}+\mathrm{const.}\left(  \frac{1}{t^{\frac{5}{2}}}\bar{\mu}_{\alpha
}+\frac{1}{t^{\frac{3}{2}}}\tilde{\mu}_{\alpha}\right)  ~,
\end{align*}
which shows that $\kappa\partial_{k}\hat{\omega}_{2,1,0}\in\mathcal{D}%
_{\alpha-1,\frac{5}{2},\frac{3}{2}}^{1}$.

The bound on the function $\kappa\partial_{k}\hat{\omega}_{2,2,0}$ uses
(\ref{eq:bnddkf20}), Proposition~\ref{prop:sgk3} and (\ref{eq:mutomutilde}),
leading to%
\begin{align*}
\left\vert \kappa\partial_{k}\hat{\omega}_{2,2,0}\right\vert  &  =\left\vert
\frac{1}{2}e^{-\kappa\tau}\int_{t}^{\infty}\kappa\partial_{k}\check{f}%
_{2,0}(k,s-1)\hat{Q}_{0}(k,s)ds\right\vert \\
&  \leq\mathrm{const.}~e^{\Lambda_{-}\tau}e^{|k|\tau}\int_{t}^{\infty
}(|k|^{\frac{1}{2}}+|k|^{2})\sigma e^{-|k|\sigma}\left(  \frac{1}{s^{\frac
{7}{2}}}\bar{\mu}_{\alpha}+\frac{1}{s^{\frac{5}{2}}}\tilde{\mu}_{\alpha
}\right)  ds\\
&  \leq\mathrm{const.}~(1+|k|)e^{\Lambda_{-}\tau}\left(  \frac{1}{t^{2}}%
\bar{\mu}_{\alpha}+\frac{1}{t^{1}}\tilde{\mu}_{\alpha}\right)  \leq
\mathrm{const.}~(1+|k|)\frac{1}{t^{1}}\tilde{\mu}_{\alpha}~,
\end{align*}
which shows that $\kappa\partial_{k}\hat{\omega}_{2,2,0}\in\mathcal{D}%
_{\alpha-1,\infty,1}^{1}$.

The bound on the function $\kappa\partial_{k}\hat{\omega}_{2,3,0}$ uses
(\ref{eq:bnddkf30}) and Proposition~\ref{prop:sgL3}, leading to%
\begin{align*}
\left\vert \kappa\partial_{k}\hat{\omega}_{2,3,0}\right\vert  &  =\left\vert
\frac{1}{2}(e^{\kappa\tau}-e^{-\kappa\tau})\int_{t}^{\infty}\kappa\partial
_{k}\check{f}_{3,0}(k,s-1)\hat{Q}_{0}(k,s)ds\right\vert \\
&  \leq\mathrm{const.}~e^{|\Lambda_{-}|\tau}\int_{t}^{\infty}(1+|\Lambda
_{-}|\sigma)e^{\Lambda_{-}\sigma}\left(  \frac{1}{s^{\frac{7}{2}}}\bar{\mu
}_{\alpha}+\frac{1}{s^{\frac{5}{2}}}\tilde{\mu}_{\alpha}\right)  ds\\
&  \leq\mathrm{const.}~e^{|\Lambda_{-}|\tau}\int_{t}^{\infty}e^{\Lambda
_{-}\sigma}\left(  \frac{1}{s^{\frac{7}{2}}}\bar{\mu}_{\alpha}+\frac
{1}{s^{\frac{5}{2}}}\tilde{\mu}_{\alpha}\right)  ds\\
&  +\mathrm{const.}~e^{|\Lambda_{-}|\tau}\int_{t}^{\infty}|\Lambda
_{-}|e^{\Lambda_{-}\sigma}\left(  \frac{1}{s^{\frac{5}{2}}}\bar{\mu}_{\alpha
}+\frac{1}{s^{\frac{3}{2}}}\tilde{\mu}_{\alpha}\right)  ds\\
&  \leq\mathrm{const.}\left(  \frac{1}{t^{\frac{5}{2}}}\bar{\mu}_{\alpha
}+\frac{1}{t^{\frac{3}{2}}}\tilde{\mu}_{\alpha}\right)  ~,
\end{align*}
which shows that $\kappa\partial_{k}\hat{\omega}_{2,3,0}\in\mathcal{D}%
_{\alpha-1,\frac{5}{2},\frac{3}{2}}^{1}$.

The bound on the function $\kappa\partial_{k}\hat{\omega}_{2,1,1}$ uses
(\ref{eq:bnddkf11}) and Propositions~\ref{prop:sgL1} and \ref{prop:sgL2},
leading to%
\begin{align*}
\left\vert \kappa\partial_{k}\hat{\omega}_{2,1,1}\right\vert  &  =\left\vert
\frac{1}{2}e^{-\kappa\tau}\int_{1}^{t}\kappa\partial_{k}\check{f}%
_{1,1}(k,s-1)\hat{Q}_{1}(k,s)ds\right\vert \\
&  \leq\mathrm{const.}~e^{\Lambda_{-}\tau}\int_{1}^{t}(1+|\Lambda_{-}%
|^{2})\sigma e^{|\Lambda_{-}|\sigma}\left(  \frac{1}{s^{\frac{7}{2}}}\bar{\mu
}_{\alpha}+\frac{1}{s^{\frac{5}{2}}}\tilde{\mu}_{\alpha}\right)  ds\\
&  \leq\mathrm{const.}\left(  \tilde{\mu}_{\alpha}+\frac{1}{t^{\frac{3}{2}}%
}\bar{\mu}_{\alpha}+\frac{1}{t^{\frac{1}{2}}}\tilde{\mu}_{\alpha}\right) \\
&  +\mathrm{const.}~\frac{1}{t^{1}}\tilde{\mu}_{\alpha}+\mathrm{const.}%
~|\Lambda_{-}|\left(  \frac{1}{t^{\frac{5}{2}}}\bar{\mu}_{\alpha}+\frac
{1}{t^{\frac{3}{2}}}\tilde{\mu}_{\alpha}\right)  ~,
\end{align*}
which shows that $\kappa\partial_{k}\hat{\omega}_{2,1,1}\in\mathcal{D}%
_{\alpha-1,\frac{3}{2},0}^{1}$.

The bound on the function $\kappa\partial_{k}\hat{\omega}_{2,2,1}$ uses
(\ref{eq:bnddkf21}), Proposition~\ref{prop:sgk3} and (\ref{eq:mutomutilde}),
leading to%
\begin{align*}
\left\vert \kappa\partial_{k}\hat{\omega}_{2,2,1}\right\vert  &  =\left\vert
\frac{1}{2}e^{-\kappa\tau}\int_{t}^{\infty}\kappa\partial_{k}\check{f}%
_{2,1}(k,s-1)\hat{Q}_{1}(k,s)ds\right\vert \\
&  \leq\mathrm{const.}~e^{\Lambda_{-}\tau}e^{|k|\tau}\int_{t}^{\infty
}(1+|k|^{2})\sigma e^{-|k|\sigma}\left(  \frac{1}{s^{\frac{7}{2}}}\bar{\mu
}_{\alpha}+\frac{1}{s^{\frac{5}{2}}}\tilde{\mu}_{\alpha}\right)  ds\\
&  \leq\mathrm{const.}~(1+|k|)e^{\Lambda_{-}\tau}\left(  \frac{1}{t^{\frac
{3}{2}}}\bar{\mu}_{\alpha}+\frac{1}{t^{\frac{1}{2}}}\tilde{\mu}_{\alpha
}\right)  \leq\mathrm{const.}~(1+|k|)\frac{1}{t^{\frac{1}{2}}}\tilde{\mu
}_{\alpha}~,
\end{align*}
which shows that $\kappa\partial_{k}\hat{\omega}_{2,2,1}\in\mathcal{D}%
_{\alpha-1,\infty,\frac{1}{2}}^{1}$.

The bound on the function $\kappa\partial_{k}\hat{\omega}_{2,3,1}$ uses
(\ref{eq:bnddkf31}) and Proposition~\ref{prop:sgL3}, leading to%
\begin{align*}
\left\vert \kappa\partial_{k}\hat{\omega}_{2,3,1}\right\vert  &  =\left\vert
\frac{1}{2}(e^{\kappa\tau}-e^{-\kappa\tau})\int_{t}^{\infty}\kappa\partial
_{k}\check{f}_{3,1}\hat{Q}_{1}(k,s)ds\right\vert \\
&  \leq\mathrm{const.}~(e^{|\Lambda_{-}|\tau}+e^{\Lambda_{-}\tau})\int
_{t}^{\infty}(1+|\Lambda_{-}|)\sigma e^{\Lambda_{-}\sigma}\left(  \frac
{1}{s^{\frac{7}{2}}}\bar{\mu}_{\alpha}+\frac{1}{s^{\frac{5}{2}}}\tilde{\mu
}_{\alpha}\right)  ds\\
&  \leq\mathrm{const.}\left(  \frac{1}{t^{\frac{3}{2}}}\bar{\mu}_{\alpha
}+\frac{1}{t^{\frac{1}{2}}}\tilde{\mu}_{\alpha}\right)  ~,
\end{align*}
which shows that $\kappa\partial_{k}\hat{\omega}_{2,3,1}\in\mathcal{D}%
_{\alpha-1,\frac{3}{2},\frac{1}{2}}^{1}$.

Collecting the bounds we have that $\hat{d}_{2}\in\mathcal{D}_{\alpha
-1,\frac{3}{2},0}^{1}$~, which completes the second part of the proof of
Proposition~\ref{prop:d1&d2space}.

\subsection{\label{sec:boundsd3}Bounds on $\hat{d}_{3}$}

We prove the bounds on $\hat{d}_{3}$ needed to complete the proof of
Lemma~\ref{lem:Q}. For compatibility with the maps $\mathfrak{L}_{1}$ and
$\mathfrak{L}_{2}$ we will bound $\kappa\hat{d}_{3}$ instead of $\hat{d}_{3}$.
Throughout this proof we will use without further mention the bounds%
\begin{align*}
\left\vert \partial_{k}\hat{Q}_{0}\left(  k,s\right)  \right\vert  &
\leq\left\Vert \partial_{k}\hat{Q}_{0}\right\Vert \left(  \frac{1}{s^{\frac
{3}{2}}}\bar{\mu}_{\alpha}+\frac{1}{s^{1}}\tilde{\mu}_{\alpha}\right)  ~,\\
\left\vert \partial_{k}\hat{Q}_{1}\left(  k,s\right)  \right\vert  &
\leq\left\Vert \partial_{k}\hat{Q}_{1}\right\Vert \left(  \frac{1}{s^{\frac
{3}{2}}}\bar{\mu}_{\alpha}+\frac{1}{s^{2}}\tilde{\mu}_{\alpha}\right)  ~.
\end{align*}
The bound on the function $\kappa\partial_{k}\hat{\omega}_{3,1,0}$ uses
(\ref{eq:bndf10}) and Propositions~\ref{prop:sgL1} and \ref{prop:sgL2},
leading to%
\begin{align*}
\left\vert \kappa\partial_{k}\hat{\omega}_{3,1,0}\right\vert  &  =\left\vert
\frac{1}{2}e^{-\kappa\tau}\int_{1}^{t}\check{f}_{1,0}\left(  k,\sigma\right)
\kappa\partial_{k}\hat{Q}_{0}\left(  k,s\right)  ds\right\vert \\
&  \leq\mathrm{const.}~|\Lambda_{-}|e^{\Lambda_{-}\tau}\int_{1}^{t}%
e^{|\Lambda_{-}|\sigma}\min\{|\Lambda_{-}|,|\Lambda_{-}|^{3}\sigma
^{2}\}\left(  \frac{1}{s^{\frac{3}{2}}}\bar{\mu}_{\alpha}+\frac{1}{s}%
\tilde{\mu}_{\alpha}\right)  ds\\
&  \leq\mathrm{const.}~|\Lambda_{-}|e^{\Lambda_{-}\tau}\int_{1}^{\frac{t+1}%
{2}}e^{|\Lambda_{-}|\sigma}|\Lambda_{-}|^{3}\sigma^{2}\left(  \frac
{1}{s^{\frac{3}{2}}}\bar{\mu}_{\alpha}+\frac{1}{s}\tilde{\mu}_{\alpha}\right)
ds\\
&  +\mathrm{const.}~|\Lambda_{-}|e^{\Lambda_{-}\tau}\int_{\frac{t+1}{2}}%
^{t}e^{|\Lambda_{-}|\sigma}|\Lambda_{-}|\left(  \frac{1}{s^{\frac{3}{2}}}%
\bar{\mu}_{\alpha}+\frac{1}{s}\tilde{\mu}_{\alpha}\right)  ds\\
&  \leq\mathrm{const.}~|\Lambda_{-}|\left(  \frac{1}{t^{3}}\tilde{\mu}%
_{\alpha}+\frac{1}{t^{\frac{3}{2}}}\bar{\mu}_{\alpha}+\frac{1}{t^{1}}%
\tilde{\mu}_{\alpha}\right)  ~,
\end{align*}
which shows that $\kappa\partial_{k}\hat{\omega}_{3,1,0}\in\mathcal{D}%
_{\alpha-1,\frac{3}{2},1}^{1}$.

The bound on the function $\kappa\partial_{k}\hat{\omega}_{3,2,0}$ uses
(\ref{eq:bndf20}), (\ref{eq:mutomutilde}) and Proposition~\ref{prop:sgk3},
which, to be applicable, requires first the use of
(\ref{eq:ksacrificealphafort}) to trade a $|k|$ for an $s^{-1}$ multiplying
$\bar{\mu}_{\alpha}$ and $\tilde{\mu}_{\alpha}$. We then have%
\begin{align*}
\left\vert \kappa\partial_{k}\hat{\omega}_{3,2,0}\right\vert  &  =\left\vert
\frac{1}{2}e^{-\kappa\tau}\int_{t}^{\infty}\check{f}_{2,0}(k,s-1)\kappa
\partial_{k}\hat{Q}_{0}(k,s)ds\right\vert \\
&  \leq\mathrm{const.}~e^{\Lambda_{-}\tau}\int_{t}^{\infty}(|k|+|k|^{\frac
{1}{2}})(|k|^{\frac{1}{2}}+|k|)e^{-|k|\sigma}\left(  \frac{1}{s^{\frac{3}{2}}%
}\bar{\mu}_{\alpha}+\frac{1}{s}\tilde{\mu}_{\alpha}\right)  ds\\
&  \leq\mathrm{const.}~e^{\Lambda_{-}\tau}e^{|k|\tau}\int_{t}^{\infty
}|k|e^{-|k|\sigma}\left(  \frac{1}{s^{\frac{3}{2}}}\bar{\mu}_{\alpha}+\frac
{1}{s^{\frac{5}{2}}}\bar{\mu}_{\alpha-1}\right)  ds\\
&  +\mathrm{const.}~e^{\Lambda_{-}\tau}e^{|k|\tau}\int_{t}^{\infty}\left(
1+|k|\right)  e^{-|k|\sigma}\frac{1}{s^{3}}\tilde{\mu}_{\alpha-1}ds\\
&  \leq\mathrm{const.}~e^{\Lambda_{-}\tau}\left(  \frac{1}{t^{\frac{3}{2}}%
}\bar{\mu}_{\alpha}+\frac{1}{t^{\frac{5}{2}}}\bar{\mu}_{\alpha-1}+\frac
{1}{t^{2}}\tilde{\mu}_{\alpha-1}\right) \\
&  \leq\mathrm{const.}\left(  \frac{1}{t^{\frac{3}{2}}}\tilde{\mu}_{\alpha
}+\frac{1}{t^{2}}\tilde{\mu}_{\alpha-1}\right)  ~,
\end{align*}
which shows that $\kappa\partial_{k}\hat{\omega}_{3,2,0}\in\mathcal{D}%
_{\alpha-1,\infty,\frac{3}{2}}^{1}$.

The bound on the function $\kappa\partial_{k}\hat{\omega}_{3,3,0}$ uses
(\ref{eq:bndf30}) and Proposition~\ref{prop:sgL3}, which, to be applicable,
requires first the use of (\ref{eq:ksacrificealphafort}) to trade a
$|\Lambda_{-}|$ for a $s^{-1/2}$ multiplying $\tilde{\mu}_{\alpha}$. We then
have%
\begin{align*}
\left\vert \kappa\partial_{k}\hat{\omega}_{3,3,0}\right\vert  &  =\left\vert
\frac{1}{2}(e^{\kappa\tau}-e^{-\kappa\tau})\int_{t}^{\infty}\check{f}%
_{3,0}(k,s-1)\kappa\partial_{k}\hat{Q}_{0}(k,s)ds\right\vert \\
&  \leq\mathrm{const.}~e^{|\Lambda_{-}|\tau}\int_{t}^{\infty}\min
\{1,|\Lambda_{-}|\}e^{\Lambda_{-}\sigma}|\Lambda_{-}|\left(  \frac{1}%
{s^{\frac{3}{2}}}\bar{\mu}_{\alpha}+\frac{1}{s}\tilde{\mu}_{\alpha}\right)
ds\\
&  \leq\mathrm{const.}~e^{|\Lambda_{-}|\tau}\int_{t}^{\infty}|\Lambda
_{-}|e^{\Lambda_{-}\sigma}\left(  \frac{1}{s^{\frac{3}{2}}}\bar{\mu}_{\alpha
}+\frac{1}{s^{2}}\tilde{\mu}_{\alpha-\frac{1}{2}}+\frac{1}{s^{3}}\tilde{\mu
}_{\alpha-1}\right)  ds\\
&  \leq\mathrm{const.}\left(  \frac{1}{t^{\frac{3}{2}}}\bar{\mu}_{\alpha
}+\frac{1}{t^{2}}\tilde{\mu}_{\alpha-\frac{1}{2}}+\frac{1}{t^{3}}\tilde{\mu
}_{\alpha-1}\right)  ~,
\end{align*}
which shows that $\kappa\partial_{k}\hat{\omega}_{3,3,0}\in\mathcal{D}%
_{\alpha-1,\frac{3}{2},\frac{3}{2}}^{1}$.

The bound on the function $\kappa\partial_{k}\hat{\omega}_{3,1,1}$ uses
(\ref{eq:bndf11}) and Propositions~\ref{prop:sgL1} and \ref{prop:sgL2},
leading to%
\begin{align*}
\left\vert \kappa\partial_{k}\hat{\omega}_{3,1,1}\right\vert  &  =\left\vert
\frac{1}{2}e^{-\kappa\tau}\int_{1}^{t}\check{f}_{1,1}(k,s-1)\kappa\partial
_{k}\hat{Q}_{1}(k,s)ds\right\vert \\
&  \leq\mathrm{const.}~e^{\Lambda_{-}\tau}\int_{1}^{t}(1+|\Lambda
_{-}|)e^{|\Lambda_{-}|\sigma}\min\{1,|\Lambda_{-}|\sigma\}|\Lambda_{-}|\left(
\frac{1}{s^{\frac{3}{2}}}\bar{\mu}_{\alpha}+\frac{1}{s^{2}}\tilde{\mu}%
_{\alpha}\right)  ds\\
&  \leq\mathrm{const.}~(1+|\Lambda_{-}|)e^{\Lambda_{-}\tau}\int_{1}%
^{\frac{t+1}{2}}|\Lambda_{-}|e^{|\Lambda_{-}|\sigma}|\Lambda_{-}|\sigma\left(
\frac{1}{s^{\frac{3}{2}}}\bar{\mu}_{\alpha}+\frac{1}{s^{2}}\tilde{\mu}%
_{\alpha}\right)  ds\\
&  +\mathrm{const.}~(1+|\Lambda_{-}|)e^{\Lambda_{-}\tau}\int_{\frac{t+1}{2}%
}^{t}|\Lambda_{-}|e^{|\Lambda_{-}|\sigma}\left(  \frac{1}{s^{\frac{3}{2}}}%
\bar{\mu}_{\alpha}+\frac{1}{s^{2}}\tilde{\mu}_{\alpha}\right)  ds\\
&  \leq\mathrm{const.}~(1+|\Lambda_{-}|)\left(  \frac{1}{t^{\frac{3}{2}}%
}\tilde{\mu}_{\alpha}+\frac{1}{t^{\frac{3}{2}}}\bar{\mu}_{\alpha}+\frac
{1}{t^{2}}\tilde{\mu}_{\alpha}\right)  ~,
\end{align*}
which shows that $\kappa\partial_{k}\hat{\omega}_{3,1,1}\in\mathcal{D}%
_{\alpha-1,\frac{3}{2},\frac{3}{2}}^{1}$.

The bound on the function $\kappa\partial_{k}\hat{\omega}_{3,2,1}$ uses
(\ref{eq:bndf21}), Proposition~\ref{prop:sgk3} and (\ref{eq:mutomutilde}),
leading to%
\begin{align*}
\left\vert \kappa\partial_{k}\hat{\omega}_{3,2,1}\right\vert  &  =\left\vert
\frac{1}{2}e^{-\kappa\tau}\int_{t}^{\infty}\check{f}_{2,1}(k,s-1)\kappa
\partial_{k}\hat{Q}_{1}(k,s)ds\right\vert \\
&  \leq\mathrm{const.}~e^{\Lambda_{-}\tau}\int_{t}^{\infty}(1+|k|)|\Lambda
_{-}|e^{-|k|\sigma}\left(  \frac{1}{s^{\frac{3}{2}}}\bar{\mu}_{\alpha}%
+\frac{1}{s^{2}}\tilde{\mu}_{\alpha}\right)  ds\\
&  \leq\mathrm{const.}~(1+|k|)e^{\Lambda_{-}\tau}e^{|k|\tau}\int_{t}^{\infty
}(|k|^{\frac{1}{2}}+|k|)e^{-|k|\sigma}\left(  \frac{1}{s^{\frac{3}{2}}}%
\bar{\mu}_{\alpha}+\frac{1}{s^{2}}\tilde{\mu}_{\alpha}\right)  ds\\
&  \leq\mathrm{const.}~(1+|k|)e^{\Lambda_{-}\tau}\left(  \frac{1}{t^{1}}%
\bar{\mu}_{\alpha}+\frac{1}{t^{\frac{3}{2}}}\tilde{\mu}_{\alpha}\right)
\leq\mathrm{const.}~(1+|k|)\frac{1}{t^{1}}\tilde{\mu}_{\alpha}~,
\end{align*}
which shows that $\kappa\partial_{k}\hat{\omega}_{3,2,1}\in\mathcal{D}%
_{\alpha-1,\infty,1}^{1}$.

The bound on the function $\kappa\partial_{k}\hat{\omega}_{3,3,1}$ uses
(\ref{eq:bndf31}) and Proposition~\ref{prop:sgL3}, leading to%
\begin{align*}
\left\vert \kappa\partial_{k}\hat{\omega}_{3,3,1}\right\vert  &  =\left\vert
\frac{1}{2}(e^{\kappa\tau}-e^{-\kappa\tau})\int_{t}^{\infty}\check{f}%
_{3,1}\kappa\partial_{k}\hat{Q}_{1}(k,s)ds\right\vert \\
&  \leq\mathrm{const.}~(e^{|\Lambda_{-}|\tau}+e^{\Lambda_{-}\tau})\int
_{t}^{\infty}e^{\Lambda_{-}\sigma}\left\vert \Lambda_{\_}\right\vert \left(
\frac{1}{s^{\frac{3}{2}}}\bar{\mu}_{\alpha}+\frac{1}{s^{2}}\tilde{\mu}%
_{\alpha}\right)  ds\\
&  \leq\mathrm{const.}\left(  \frac{1}{t^{\frac{3}{2}}}\bar{\mu}_{\alpha
}+\frac{1}{t^{2}}\tilde{\mu}_{\alpha}\right)  ~,
\end{align*}
which shows that $\kappa\partial_{k}\hat{\omega}_{3,3,1}\in\mathcal{D}%
_{\alpha-1,\frac{3}{2},2}^{1}$.

Collecting the bounds we have that $\hat{d}_{3}\in\mathcal{D}_{\alpha
-1,\frac{3}{2},1}^{1}\subset\mathcal{D}_{\alpha-1,\frac{3}{2},0}^{1}$, which
proves Lemma~\ref{lem:Q}.

\appendix

\section{\label{app:propositions}Convolution with the semi-groups
$e^{\Lambda_{-}t}$ and $e^{-|k|t}$}

To make this paper self-contained, we recall the following results proved in
\cite{Hillairet.Wittwer-Existenceofstationary2009}. In order to bound the
integrals over the interval $[1,t]$ we systematically split them into
integrals over $[1,\frac{t+1}{2}]$ and integrals over $[\frac{t+1}{2},t]$ and
bound the resulting terms separately. For the semi-group $e^{\Lambda_{-}t}$ we have:

\begin{proposition}
\label{prop:sgL1}Let $\alpha\geq0$, $r\geq0$ and $\delta\geq0$ and
$\gamma+1\geq\beta\geq0$. Then,%
\begin{align*}
&  e^{\Lambda_{-}(t-1)}\int_{1}^{\frac{t+1}{2}}e^{|\Lambda_{-}|(s-1)}%
|\Lambda_{-}|^{\beta}\frac{(s-1)^{\gamma}}{s^{\delta}}\mu_{\alpha,r}(k,s)~ds\\
&  \leq\left\{
\begin{array}
[c]{l}%
\displaystyle\mathrm{const.}~\frac{1}{t^{\beta}}\tilde{\mu}_{\alpha
}(k,t),\text{ if }\delta>\gamma+1\\
\\
\displaystyle\mathrm{const.}~\frac{\log(1+t)}{t^{\beta}}\tilde{\mu}_{\alpha
}(k,t),\text{ if }\delta=\gamma+1\\
\\
\displaystyle\mathrm{const.}~\frac{t^{\gamma+1-\delta}}{t^{\beta}}\tilde{\mu
}_{\alpha}(k,t),\text{ if }\delta<\gamma+1
\end{array}
\right.
\end{align*}
uniformly in $t\geq1$ and $k\in\mathbb{R}$.
\end{proposition}

\begin{proposition}
\label{prop:sgL2}Let $\alpha\geq0$, $r\geq0$, $\delta\in\mathbb{R}$, and
$\beta\in\{0,1\}$. Then,%
\[
e^{\Lambda_{-}(t-1)}\int_{\frac{t+1}{2}}^{t}e^{|\Lambda_{-}|(s-1)}|\Lambda
_{-}|^{\beta}\frac{1}{s^{\delta}}\mu_{\alpha,r}(k,s)~ds\leq\frac
{\mathrm{const.}}{t^{\delta-1+\beta}}\mu_{\alpha,r}(k,t)~,
\]
uniformly in $t\geq1$ and $k\in\mathbb{R}$.
\end{proposition}

For the integral over the interval $[t,\infty)$ we need only one of the bounds
in \cite{Hillairet.Wittwer-Existenceofstationary2009}.

\begin{proposition}
\label{prop:sgL3}Let $\alpha\geq0$, $r\geq0$, $\delta>1$, and $\beta
\in\{0,1\}$. Then,%
\[
e^{|\Lambda_{-}|(t-1)}\int_{t}^{\infty}e^{\Lambda_{-}(s-1)}|\Lambda
_{-}|^{\beta}\frac{1}{s^{\delta}}\mu_{\alpha,r}(k,s)~ds\leq\frac
{\mathrm{const.}}{t^{\delta-1+\beta}}\mu_{\alpha,r}(k,t)~,
\]
uniformly in $t\geq1$ and $k\in\mathbb{R}$.
\end{proposition}

\bigskip

For the semi-group $e^{-|k|t}$ we have:

\begin{proposition}
\label{prop:sgk3}Let $\alpha\geq0$, $r\geq0$, $\delta>1$, $\beta\in
\lbrack0,1]$ Then,%
\[
e^{|k|(t-1)}\int_{t}^{\infty}e^{-|k|(s-1)}|k|^{\beta}\frac{1}{s^{\delta}}%
\mu_{\alpha,r}(k,s)~ds\leq\frac{\mathrm{const.}}{t^{\delta-1+\beta}}%
\mu_{\alpha,r}(k,t)~,
\]
uniformly in $t\geq1$ and $k\in\mathbb{R}$.
\end{proposition}

\bibliographystyle{amsplain}
\bibliography{CompleteDatabase2011}

\end{document}